\documentclass[onecolumn,draftclsnofoot]{IEEEtran}

\usepackage{ifpdf}

\usepackage{cite}

\ifCLASSINFOpdf
\usepackage[pdftex]{graphicx}
\else
\usepackage[dvips]{graphicx}
\fi
\usepackage{amsmath}
\usepackage{amssymb}
\usepackage{mathtools}
\usepackage{amsfonts}
\usepackage{dsfont}

\usepackage{algorithmic}

\usepackage{array}

\usepackage{arydshln}
\usepackage{multirow}

\ifCLASSOPTIONcompsoc
\usepackage[caption=false,font=normalsize,labelfont=sf,textfont=sf]{subfig}
\else
\usepackage[caption=false,font=footnotesize]{subfig}
\usepackage{dblfloatfix}
\usepackage{url}
\usepackage{amsthm}
\usepackage{xcolor}
\definecolor{darkgreen}{rgb}{0.2, 0.5, 0.2}
\definecolor{darkblue}{rgb}{0.2,0.2, 0.8}
\usepackage{hyperref}

\hypersetup{
	colorlinks,
	citecolor=darkgreen,
	filecolor=darkblue,
	linkcolor=darkblue,
	urlcolor=darkblue
}

\usepackage{mathrsfs}
\usepackage{tikz}
\usetikzlibrary{matrix,positioning,fit,calc,patterns}
\usepackage{pgfplots}

\usepackage{extarrows}
\usepackage[normalem]{ulem}

\newcommand{\mcal}[1]{\mathcal{#1}}
\newcommand{\intv}[1]{\left[#1\right]}
\newcommand {\sen}[1]{\mcal {#1}}
\newcommand {\seq}[1]{\set{#1}}
\newcommand{\set}[1]{\left\lbrace #1 \right\rbrace}
\newcommand{\size}[1]{ \left| #1 \right|}

\newcommand{\ind}[2]{\mathsf{ind}_{#1} (#2)}

\newtheorem{rem}{Remark}
\newtheorem{defi}{Definition}
\newtheorem{prop}{Proposition}

\newtheorem{thm}{Theorem}

\newtheorem{remark}{Remark}

\newcommand{\bX}{\mathbf{X}}

\newcommand{\cC}{\mathcal{C}}

\newcommand{\cE}{\mathcal{E}}

\newcommand{\cH}{\mathcal{H}}

\newcommand{\cV}{\mathcal{V}}
\newcommand{\cW}{\mathcal{W}}

\newcommand{\enc}{\mathbf{\Psi}}

\newcommand{\hlp}{\sen H}

\newcommand{\SM}{\mathbf{D}}

\newcommand{\repMat}[2]{\mathbf{\Xi}^{#1,(#2)}}

\newcommand{\repFam}[2]{\mathbf{\Xi}^{#1,(#2)}}

\newcommand{\NSpc}[2]{\mathbf{Y}^{#1,(#2)}}
\newcommand{\repSpc}[2]{\mathbf{R}^{#1,(#2)}}

\renewcommand{\det}[1]{\mathsf{det}\left(#1 \right)}

\usepackage{tikz}
\usetikzlibrary{matrix,positioning,fit,calc,decorations.pathreplacing}
\colorlet{colIndColor}{blue}
\colorlet{rowIndColor}{red}
\colorlet{repIndColor}{darkgreen}
\colorlet{nulIndColor}{magenta}
\colorlet{myblue}{blue!20}
\colorlet{myred}{red!35}

\definecolor{SecondInj}{RGB}{200,255,200}
\definecolor{FirstInj}{RGB}{180,230,180}

\definecolor{mygreen}{RGB}{115,225,100}
\pgfdeclarepatternformonly{my crosshatch dots}{\pgfqpoint{-1pt}{-1pt}}{\pgfqpoint{5pt}{5pt}}{\pgfqpoint{6pt}{6pt}}%
{
    \pgfpathcircle{\pgfqpoint{0pt}{0pt}}{.7pt}
    \pgfpathcircle{\pgfqpoint{3pt}{3pt}}{.7pt}
    \pgfusepath{fill}
}
\tikzset{holeBox/.style= {preaction={ fill,#1},pattern= my crosshatch dots,pattern color=white},}
\tikzset{My Style/.style={red, draw=blue, fill=yellow!20, minimum size=0.5cm}}

\definecolor{vColor}{RGB}{255,255,255}
\tikzset{
	braces/.style = {
		outer sep=-1pt,
		left delimiter=[,
		right delimiter=],
		align=center,
	},
}
\colorlet{PColor}{blue}
\colorlet{QColor}{blue}
\colorlet{ZColor}{red}
\colorlet{1Color}{orange}
\colorlet{2Color}{magenta}
\colorlet{3Color}{cyan}
\definecolor{4Color}{RGB}{100,140,0}
\colorlet{5Color}{yellow}

\newcommand{\bLozenge}{:} 

\begin{document}
    \title{Determinant Codes with Helper-Independent Repair  for Single and Multiple Failures    
}
    \author{Mehran Elyasi, and Soheil~Mohajer,~\IEEEmembership{Member,~IEEE}
\thanks{M. Elyasi and 
S. Mohajer are with the Department of Electrical and Computer Engineering, University of Minnesota, Twin Cities, MN 55455, USA, (email: \{melyasi, soheil\}@umn.edu).}}
  \maketitle

\begin{abstract}
Determinant codes are a class of exact-repair regenerating codes for distributed storage systems with parameters $(n,k=d,d)$. These codes cover the entire trade-off between per-node storage and repair-bandwidth. In an earlier work of the authors, the repair data of the determinant code sent by a helper node to repair a failed node  depends on the identity of the other helper nodes participating in the process, which is practically undesired. In this work, a new repair mechanism is proposed for determinant codes, which relaxes this dependency, while preserving all other properties of the code. 
Moreover, it is shown that the determinant codes are capable of repairing  multiple failures, with a per-node repair-bandwidth which scales sub-linearly with the number of failures. 
\end{abstract}

\section{Introduction}

While individual storage units in  distributed storage systems (DSS) are subject to temporal or permanent failure, the entire system should be designed to avoid losing the stored data. 
Coding and storing redundant data is a standard approach 
to guarantee durability in such systems. Moreover,  these systems are equipped with a repair mechanism that allows for a replacement of a failed node. Such replacement  can be performed in the \emph{functional} or \emph{exact} sense. In functional repair, a failed node will be replaced by another one, so that the consequent  family of nodes maintains the data recovery and node-repair properties. In an exact repair process, the content of a failed node will be exactly replicated by the helpers.

Regeneration codes are introduced to manage data recovery and node repair mechanism in DSS. Formally, an $(n, k, d)$ regeneration code with parameters $(\alpha,\beta,F)$ encodes a file comprised of $F$ symbols (from a finite field $\mathbb{F}$) into $n$ segments (nodes) $W_1, W_2, \dots, W_n$ each of size $\alpha$, such that two important properties are fulfilled: (1) the entire file can be recovered from  every subset of $k$ nodes, and  (2) whenever a node fails (become inaccessible), it can be repaired by accessing $d$ remaining nodes and downloading $\beta$ symbols from each.

It turns out that there is a fundamental trade-off between
the minimum required per-node storage $\alpha$ and the repair-bandwidth $\beta$, to store a given
amount of data $F$ in a DSS. This tradeoff is fully characterized for functional  repair in the seminal work of Dimakis et. al \cite{dimakis2010network}, where it is shown to be achievable by random network coding. However, for the exact-repair problem, which is notably important from the practical perspective, characterization of the trade-off and design of optimum codes are widely open, except for some special cases. Construction of exact-repair regenerating codes for a system with arbitrary parameters $(n,k,d)$ is a complex task due to several combinatorial constraints to be satisfied. The number of such constraints dramatically increases  with $n$, the total number of nodes in the system.

There are known code constructions for the two extreme points on the tradeoff, namely, minimum bandwidth regeneration (MBR) \cite{rashmi2011optimal} and minimum storage regeneration (MSR) \cite{rashmi2011optimal, lin2015unified,ye2017explicit, sasidharan2016explicit,tian2017generic,ye2017explicitnearly } points.

A lower bound for the trade-off of the exact-repair regenerating codes with parameters $(n,k=d,d)$ is presented independently by \cite{elyasi2015linear, prakash2015storage, duursma2015shortened}, under the assumption that the underlying code is linear. This lower bound could be achieved for the special case of $n=k+1$ (i.e., the DSS can only tolerate one failure) by code constructions proposed in \cite{tian2015layered}  and \cite{goparaju2014new}. While the lower bound does not depend on $n$ (the total number of nodes in the system), the performance (storage capacity) of both code constructions in \cite{tian2015layered}  and \cite{goparaju2014new} degrades as $n$ exceeds $k+1$ (see Fig.~\ref{fig:compare}).

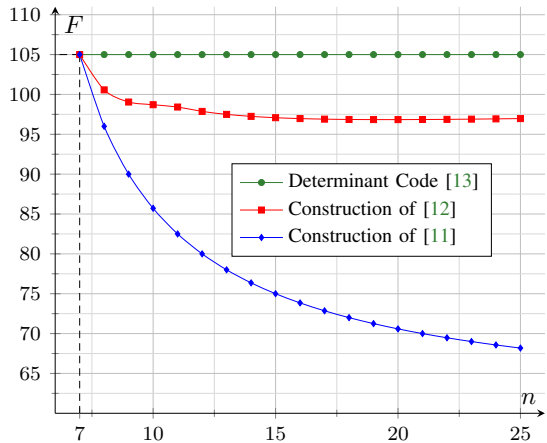
\begin{figure}
\begin{center}
	\scalebox{0.95}{
\begin{tikzpicture}
\begin{axis}[
legend style={at={(0.36,0.5)},anchor=west},
xmin=6,xmax=26,
ymin=60,ymax=111,
grid=both,
grid style={line width=.1pt, draw=gray!30},
major grid style={line width=.2pt,draw=gray!50},
axis lines=middle,
minor tick num=1,
mark options={scale=0.6},
legend cell align={left},
legend style={font=\footnotesize},
ticklabel style={font=\footnotesize,fill=white},
xlabel=$n$,
ylabel=$F$,
extra tick style={
	xmajorgrids=false,
	ymajorgrids=false,
},
extra x ticks={7},
ytick = {60,65, ..., 110},
xtick = {5,10, ..., 30},
]

\addplot[smooth,color=darkgreen,mark=*]
plot coordinates {
(7,105)
(8,105)
(9,105)
(10,105)
(11,105)
(12,105)
(13,105)
(14,105)
(15,105)
(16,105)
(17,105)
(18,105)
(19,105)
(20,105)
(21,105)
(22,105)
(23,105)
(24,105)
(25,105)
};
\addlegendentry{Determinant Code \cite{elyasi2016determinant}}

\addplot[smooth,color=red,mark=square*]
plot coordinates {
(7,105)
(8,704/7)
(9,2080/21)
(10,3060/31)
(11,1870/19)
(12,7536/77)
(13,195/2)
(14,41720/429)
(15,15630/161)
(16,3200/33)
(17,969/10)
(18,278280/2873)
(19,54910/567)
(20,1840/19)
(21,11235/116)
(22,11528/119)
(23,99130/1023)
(24,42360/437)
(25,3200/33)
};
\addlegendentry{Construction of \cite{goparaju2014new}}

\addplot[smooth,color=blue,mark=diamond*]
plot coordinates {
(7,105)
(8,96)
(9,90)
(10,600/7)
(11,165/2)
(12,80)
(13,78)
(14,840/11)
(15,75)
(16,960/13)
(17,510/7)
(18,72)
(19,285/4)
(20,1200/17)
(21,70)
(22,1320/19)
(23,69)
(24,480/7)
(25,750/11)
};
\addlegendentry{Construction of \cite{tian2015layered}}
\draw [densely dashed] (axis cs:7,105) -- (axis cs:7,0)  ;
\draw [densely dashed] (axis cs:7,105) -- (axis cs:0,105)  ;
\end{axis}
\end{tikzpicture}
}
\end{center}	
\caption{Comparison of storage capacity of three exact-regenerating codes for an $(n,k,k)$ distributed storage system, with $d=k=6$ at mode $m=3$, and code  parameters $(\alpha,\beta)= (20,10)$. All three codes can store $F=105$ units of data when DSS has only $n=d+1=7$ nodes. However, the storage capacity decays as a function of $n$ for codes introduced in \cite{goparaju2014new} and \cite{tian2015layered}, while the storage capacity is preserved for the determinant code.
}
\label{fig:compare}
\end{figure}

\subsection{Determinant Codes and Helper-Independent Repair}
Determinant codes are a family of exact repair regenerating codes, which are introduced in \cite{elyasi2016new, elyasi2016determinant} for a DSS with parameters $(n,k=d,d)$. The main property of these codes is to maintain a constant trade-off between $\alpha/F$ and $\beta/F$, regardless of the number of the nodes. In particular, these codes can achieve the lower bound \cite{elyasi2015linear, prakash2015storage, duursma2015shortened}, and hence they are optimum. The determinant codes have a linear structure and can be obtained from the inner product between an encoder matrix and the message matrix. Especially, product-matrix codes introduced in \cite{rashmi2011optimal} for MBR and MSR points can be subsumed from the general construction of the determinant codes. 

The repair mechanism proposed for the determinant codes in the original paper \cite{elyasi2016determinant} requires a rather heavy computation at the helper nodes  in order to prepare their repair symbols to send to the failed node. More importantly, each helper node $h\in \cH$ needs to know the identity of all the other helper nodes participating in the repair process.  The assumption of knowing the set of helpers in advance is a limitation of the determinant codes, and it is  undesired in  real-world systems.  In practice, it is preferable that once a request for a repair of a failed node is made, each node can independently decide to whether or not to participate in the repair process and generate the repair data from its content, regardless of the other helper nodes.

On the other hand, besides the repair bandwidth, one of the crucial bottlenecks in the performance of the storage systems is the I/O load, which refers to the amount of data to be read by a helper node to encode for a repair process. While the native constructions for exact repair generating code require a heavy I/O read, the repair-by-transfer (RBT) codes \cite{shah2012distributed} offer an optimum I/O. In \cite{rashmi2015having} an elegant modification is proposed to improve the I/O cost of product-matrix MSR codes, by pre-processing the content of the nodes and storing the repair data on non-systematic nodes instead of the original node content. This results in a semi-RBT code: whenever such modified nodes contribute in a repair process, they merely transfer some of their stored symbols without any computation. Such modification could not be applied on the original determinant codes since the repair symbols from a helper node $h$ to a failed node $f$ could be computed only when the set of other helper nodes $\cH$ is identified. 


In this paper, we propose a novel repair mechanism for the determinant codes introduced in \cite{elyasi2016determinant}. In the new repair procedure, data repair symbols from helper node $h$ to a failed node $f$ solely depend on the content of the helper node and the identity of the failed node $f$. The failed node collects a total of $d\beta$ repair symbols from the helper nodes and can reconstruct  all of its missing symbols by simple addition and subtraction of some of the received symbols. This simple repair scheme further allows for modifications proposed in \cite{rashmi2015having}, to further improve the I/O overhead of the code.

\subsection{Multiple Failures Repair}
The second contribution of this work is the simultaneous repair for multiple failures. Although single failures are the dominant type of failures in distributed storage systems \cite{rashmi2013solution}, multiple simultaneous failures occur rather frequently and need to be handled in order to maintain the system's reliability and fault-tolerance. The naive approach to deal with such failures is to repair each failed node individually and independently from the others. This requires a repair bandwidth from each helper node that scales  linearly with the number of failures. There are two types of repair for multiple failures studied in the literature \cite{zorgui2017centralized}: (i)  \emph{centralized regenerating codes} and (ii) \emph{cooperative regenerating codes}.
In centralized regenerating codes, a single data center is responsible for the repair of all failed nodes. More precisely, once a set of $e$ nodes in the  system fail, an arbitrary set of $d\leq n-e$ nodes are chosen, and  $\beta_{e}$ repair symbols will be downloaded from each helper node. This leads to  a total of $d\cdot \beta_e$ symbols which will be used to repair the content of all the failed nodes.  The storage-bandwidth trade-off of these codes are studied for two extreme points, namely the minimum storage multi-node repair (MSMR) and the minimum bandwidth multi-node repair (MBMR) points. In particular, in \cite{ye2017explicit}  a class of MSMR code is introduced, that are capable of repairing any number of failed nodes $e\leq n-k$ from any number of helper nodes $ k \leq d \leq n-e$, using an optimal repair bandwidth.  In cooperative regenerating codes upon failure of a node, the replacement node downloads repair data from a subset of $d$ helper nodes. In the case of multiple failures, the replacement nodes not only download repair data from the helper nodes, but also exchange information among themselves before regenerating the lost data, and this exchanged data between them is included in the repair bandwidth. Similar to the centralized  case, the trade-off for these codes for the two extreme points, namely the minimum bandwidth cooperative regeneration (MBCR) codes \cite{wang2013exact} and  the minimum storage cooperative regenerating (MSCR) codes \cite{shum2013cooperative,li2014cooperative, ye2018cooperative} are studied. In particular, in \cite{ye2018cooperative} authors introduced explicit constructions of MDS codes with optimal cooperative repair for all possible parameters. Also, they have shown that any MDS code with optimal repair bandwidth under the cooperative model also has optimal bandwidth under the centralized model.

In this work we show that the repair bandwidth required for multiple failures repair in determinant codes can be reduced by exploiting two facts: (i) the overlap between the repair space (linear dependency between the repair symbols) that each helper node sends to the set of failed nodes, and (ii) in the centralized repair, the data center (responsible for the repair process) can perform the repair of the  nodes in a sequential manner, and utilize already repaired nodes as helpers for the repair of the remaining failed nodes. Interestingly, using these properties we can limit the maximum (normalized) repair-bandwidth of the helper nodes to a certain fraction of $\alpha$, regardless of the number of failures. The structure of the code  allows us to analyze this overlap, and obtain a closed-form expression for the repair bandwidth. Our codes are not restricted only to the extreme points of the trade-off and can operate at any intermediate point on the optimum trade-off.  
A similar problem is studied in \cite{zorgui2017centralized}, where a class of codes is introduced to operate at the intermediate points of the trade-off, with an improved repair bandwidth for multiple failures. However, this improvement is obtained at the price of degradation of the system's storage capacity as $n$ (the total number of nodes) increases. Consequently, the resulting codes designed for two or more simultaneous failures are sub-optimum, and cannot achieve the optimum trade-off between the per-node capacity, repair bandwidth, and the overall storage capacity. One of the main advantages of our proposed code and repair mechanism is to offer a universal code, which provides a significant reduction in the repair bandwidth for multiple failures, without compromising the system performance.

The rest of this paper is organized as follows: For the sake of completeness, we first review  the achievable trade-off and the  construction of the determinant codes \cite{elyasi2016determinant} in Sections~\ref{sec:main} and Section~\ref{sec:construction}. The new encoding and decoding for the node repair are presented in Section~\ref{sec:repair}. An illustrative example is provided in Section~\ref{sec:example}, in which the core idea of data recovery and node repair are  demonstrated. The formal proofs of the properties of the proposed code are presented in Section~\ref{sec:nddproofs}.  Finally, in Section~\ref{sec:MulRep-improved} we discuss the improved repair-bandwidth for multiple failures in a centralized repair setting.

\section{Main Result}
\label{sec:main}
We start by introducing  a few symbols and notations, which are frequently used in this paper. 

\emph{Notation:} We use  $\intv{k+1:d}$ to denote the set of integer numbers $\{k+1,\dots, d\}$, and $\intv{k}=\intv{1:k}$ to represent the set $\set{1,2,...,k}$. For a set $\sen X$ and a member $x\in \sen X$, we define $\ind{\sen X}{x}=\size{\set{y\in \sen X: y\leq x}}$. We use boldface symbols to refer to matrices, and for a matrix $\mathbf{X}$, we denote its $i$-th row by $\mathbf{X}_i$ .   
 We also use the notation $\bX_{\bLozenge,j}$ to refer to the $j$-th column of $\bX$. 
Moreover, we use $\mathbf{X}[\sen A, \sen B]$ to denote a sub-matrix of $\mathbf{X}$ obtained by rows $i\in \sen A$ and columns $j\in \sen B$. Accordingly, $\mathbf{X}[\sen A, :]$ denotes the sub-matrix of $\mathbf{X}$ by stacking rows  $i\in \sen A$. 
Moreover, we may use sets to label rows and/or columns of a matrix, and hence $\mathbf{X}_{\sen I,\sen J}$ refers to an entry of matrix $\mathbf{X}$ at the row indexed by $\sen I$ and the column labeled by $\sen J$.
Finally, for a set $\sen I$, we denote the maximum entry of $\sen I$ by $\max \sen I$.

The optimum storage repair-bandwidth of the exact-repair regenerating codes for an $(n,k=d,d)$ system is a piece-wise linear function \cite{elyasi2016determinant, elyasi2016new}, which is fully characterized by its corner (intersection) points \cite{elyasi2015linear, duursma2015shortened,  prakash2015storage}. The determinant codes provide a universal construction for all corner points on the optimum trade-off curve. We assign a \emph{mode} (denoted by $m$) to each corner point, which is an integer in $\{1,2,\dots, d\}$ (from $1$ for MBR to $d$ for MSR point). The main distinction  between the result of this work and that of \cite{elyasi2016determinant, elyasi2016new} is the fact that the repair data sent by one helper node does not depend on the identity of all the other helper nodes participating the repair process. The following definition formalizes this distinction.

\begin{defi}
Consider the repair process of a failed node $f$ using a set of helper nodes $\cH$. The repair process is called \emph{helper-independent} if the repair data sent by each helper node $h\in \cH$ to the failed node $f$ only depends on $f$ and the content of node $h$ (but not the other helpers participating in the repair process). 
\end{defi}

The following theorem formally states the trade-off achievable by determinant codes. 

\begin{thm}
	For an $(n,k=d,d)$ distributed storage system and any mode $m=1,2,\dots, d$, the triple  $(\alpha,\beta,F)$ with
	\begin{align}
	\left( \alpha^{(m)}, \beta^{(m)}, F^{(m)} \right) =  \left( \binom{d}{m} 
	, \binom{d-1}{m-1}
	, m \binom{d+1}{m+1} \right)
	\label{eq:nddparam}
	\end{align} 
	 can be achieved under \emph{helper-independent exact repair} by the code construction proposed in this paper.
	 \label{thm:main}
\end{thm}

It is worth mentioning that this theorem and the achievable points on the trade-off curve are identical to those of \cite{elyasi2016determinant}. However, the novel repair process presented here has the advantage that the repair data sent by a helper node does not depend on the identity of other helpers participating in the repair process. Moreover, we present a repair mechanism for multiple simultaneous failures. The proposed scheme exploits the overlap between the repair data sent for different failed nodes and offers a reduced repair-bandwidth compared to naively repairing the failed nodes independent of each other.

The code construction is reviewed in Section~\ref{sec:construction} for completeness. In order to prove Theorem~\ref{thm:main}, it suffices to show that the proposed code satisfies the two fundamental properties, namely data recovery and exact node repair. The proof data recovery property is similar to that of \cite[Proposition~1]{elyasi2016determinant}, and hence omitted here. The exact-repair property is formally stated in  Proposition~\ref{prop:node:rep}, and proved in Section~\ref{sec:nddproofs}. Moreover, Proposition~\ref{lm:beta} shows that the repair bandwidth of the proposed code does not exceed $\beta^{(m)}$. This is also proved in Section~\ref{sec:nddproofs}.

In Fig.~\ref{fig:tradeoff}  the linear trade-off  for a system $d=4$ together with achievable corner points of this paper are depicted.

\begin{figure}[h!]
	\centering
	\begin{tikzpicture}
	\node[anchor=north east,inner sep=0] (image) at (0,0) {\includegraphics[width=0.4\linewidth]{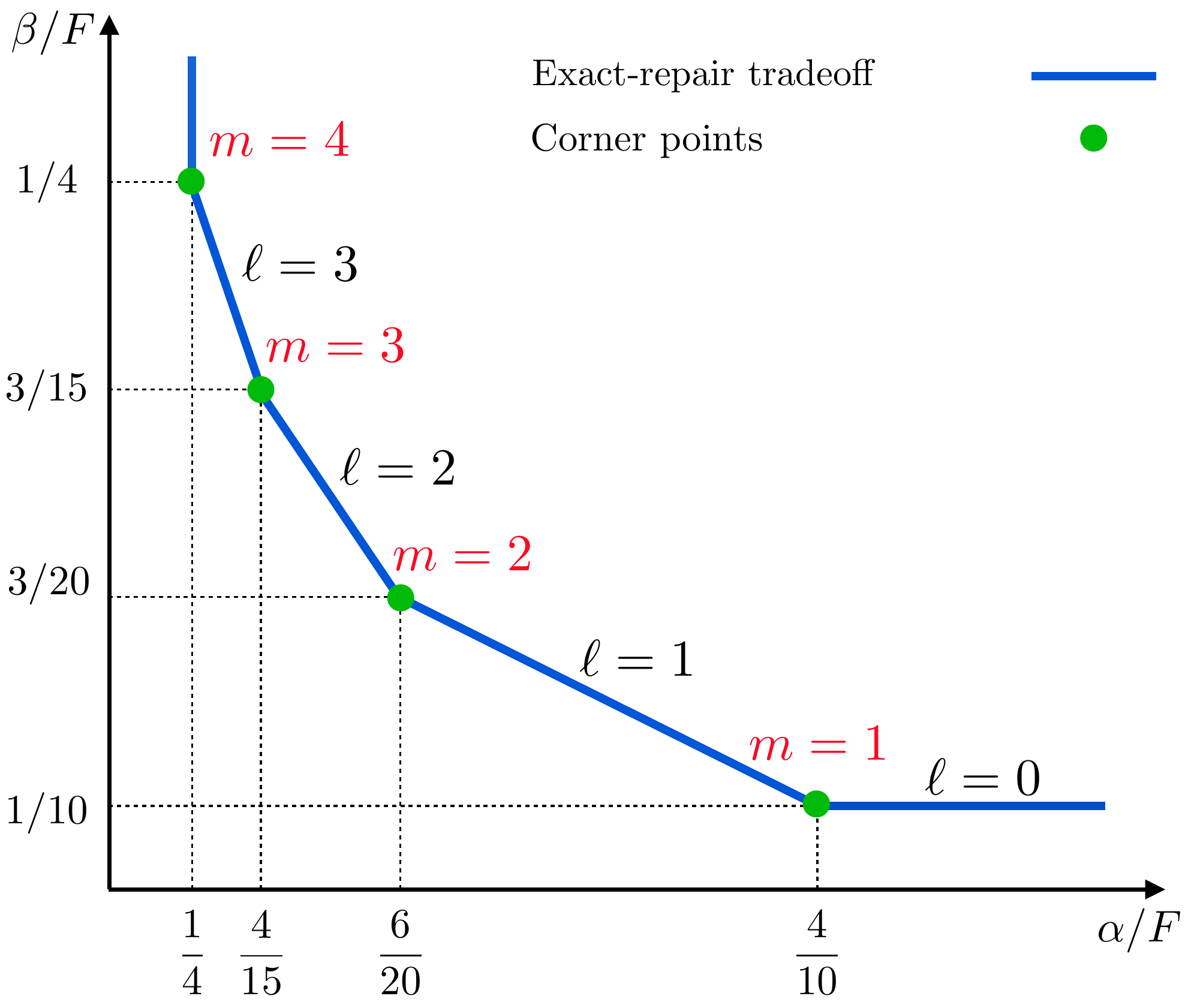}};
	\begin{scope}[x={(image.north west)},y={(image.south east)}]
	\node[fill = white, outer sep=0pt, inner sep=0pt,align = left] (rect) at (0.16,0.072) { 
		{\hspace{-1pt}\small \cite{elyasi2015linear,prakash2015storage,duursma2015shortened}}
	};	
	\end{scope}
	\end{tikzpicture}
	\vspace{-4pt}
	\caption{The optimum linear trade-off (from \cite{elyasi2016determinant}) between the normalized node-storage $\alpha/F$ and the normalized repair-bandwidth $\beta/F$ for a $(5,4,4)$-DSS given in \eqref{eq:lin-bound}. The coordinates of the corner points are given by \eqref{eq:cp}.}
	\label{fig:tradeoff}
	\vspace{-3pt}
\end{figure}

\begin{remark}
Theorem~\ref{thm:main} offers an achievable trade-off for the normalized parameters $(\alpha/F, \beta/F)$ given by
\begin{align}
	\left(\frac{\alpha^{(m)}}{F^{(m)}},\frac{\beta^{(m)}}{F^{(m)}}\right)&=\left( \binom{d}{m}/m\binom{d+1}{m+1}, \binom{d-1}{m-1}/m\binom{d+1}{m+1}\right)
	= \left(\frac{m+1}{m(d+1)}, \frac{m+1}{d(d+1)}\right). 	
	\label{eq:cp}
	\end{align} 

It is shown in \cite{elyasi2015linear,prakash2015storage,duursma2015shortened} that for any \emph{linear} exact-repair regenerating code with parameters $(n,k=d,d)$ that is capable of storing $F$ symbols, $(\alpha, \beta)$ should satisfy
	\begin{align}
	F\leq \frac{d+1}{\ell+2}\left(\ell\alpha + \frac{d}{\ell+1} \beta\right),
	\label{eq:lin-bound}
	\end{align}
	where $\ell=\lfloor d\beta/\alpha\rfloor$ takes values in $\{0,1,\dots, d\}$. This establishes a piece-wise linear lower bound curve, with $d$ (normalized) corner points obtained at integer values of $\ell = d\beta/\alpha$. For these corner points, the (normalized) operating points $(\alpha/F, \beta/F)$ are given 
	\begin{align*}
	\left(\frac{\alpha}{F}, \frac{\beta}{F}\right) = \left(\frac{\ell+1}{\ell(d+1)}, \frac{\ell+1}{d(d+1)}\right).
	\end{align*}		
	These operating points are matching with the achievable (normalized) pairs given in \eqref{eq:cp}. 
	Therefore, determinant codes are optimal, and together with the lower bound of 
	\cite{elyasi2015linear,prakash2015storage,duursma2015shortened} fully characterize the optimum trade-off for exact-repair regenerating codes with parameters $(n,k=d,d)$. 
\end{remark}
	

The next result of this paper provides an achievable bandwidth for  multiple repairs.  
\begin{thm}
In an $(n,k=d,d)$ determinant codes operating at mode $m$, the content of any set of  $e$ simultaneously failed nodes can be \emph{exactly} repaired by accessing an arbitrary set of $d$ nodes and downloading  
\begin{align*}
\beta_e^{(m)} = \binom{d}{m}-\binom{d-e}{m}
\end{align*}
repair symbols from each helper node. 
\label{thm:MulRep}
\end{thm}
The repair mechanism for multiple failures is similar to that of single failure presented in Proposition~\ref{prop:node:rep}. In order to prove Theorem~\ref{thm:MulRep}, it suffices to show that the repair bandwidth required for multiple failures does not exceed $\beta_e^{(m)}$. This is formally stated in Proposition~\ref{prop:multi-beta} and proved in Section~\ref{sec:nddproofs}. 

\begin{figure}
\begin{center}
\scalebox{0.5}{
\begin{tikzpicture}
\begin{axis}[
height=13cm,
width = 17cm,
legend style={font=\large,at={(0.45,0.25)},anchor=south west},
xmin=0.5,
 xmax=10.5,
domain=0:10,
ymin=0.25,
ymax=1.02,
mark options={scale=0.73},
grid=both,
grid style={line width=.1pt, draw=gray!30},
major grid style={line width=.2pt,draw=gray!50},
axis lines=middle,
label style={font=\LARGE},
tick label style={font=\large},
legend cell align={left},
ticklabel style={fill=white},
tension =0,
xlabel=$e$,
yticklabel style={ scaled ticks=false,
	/pgf/number format/fixed,
	/pgf/number format/precision=4},
extra y ticks={0.825},
xtick={1,2,...,10},
ylabel=$\frac{\beta_e}{\alpha}$]
\addplot[very thick,smooth,color=red,mark=square*]
plot coordinates {
(1,3/10)
(2,3/5)
(3,9/10)
(4,1)
(5,1)
(6,1)
(7,1)
(8,1)
(9,1)
(10,1)
};
\addlegendentry{Naive Multiple Repairs}
\addplot[very thick,smooth,mark=*,blue] plot coordinates {
(1,3/10)
(2,8/15)
(3,17/24)
(4,5/6)
(5,11/12)
(6,29/30)
(7,119/120)
(8,1)
(9,1)
(10,1)
};
\addlegendentry{Multiple Repairs using Theorem~\ref{thm:MulRep}}
\addplot[very thick,smooth,mark=*,darkgreen] plot coordinates {
(1,3/10)
(2,51/100)
(3,13/20)
(4,59/80)
(5,63/80)
(6,13/16)
(7,329/400)
(8,33/40)
(9,33/40)
(10,33/40)
};
\addlegendentry{Multiple Repairs using Theorem~\ref{thm:MulRep-improved}}
\addplot[very thick,mark=none, color={rgb:red,200;green,10;blue,10},thick, densely dashed] {(3*11)/(10*4)};
\addlegendentry{Saturation of $\bar{\beta_e}/\alpha$ regardless of $e$}
\end{axis}
\end{tikzpicture}
}
\end{center}
\caption{Normalized repair bandwidth (by per-node storage size) for multiple failures with $e$ failed nodes, for an $(n,10,10)$ determinant code operating at mode $m=3$, i.e, $(\alpha,\beta, F)=(120,36, 990)$.}
\label{fig:MultipleRepair}
\end{figure}
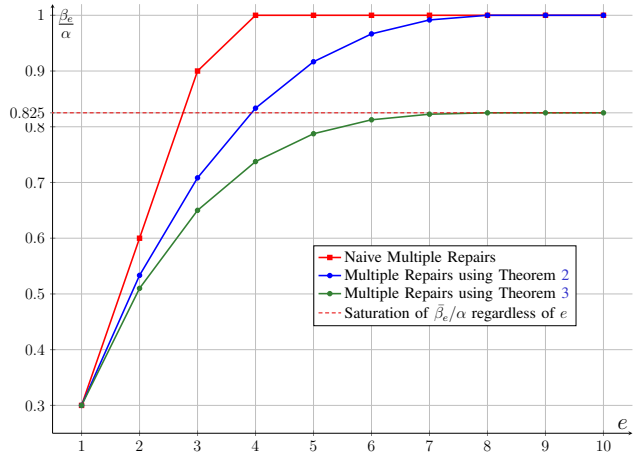

\begin{remark}
Note that the repair bandwidth proposed for multiple repairs in Theorem~\ref{thm:MulRep} subsumes the one in Theorem~\ref{thm:main} for single failure for setting $e=1$:
\begin{align*}
\beta_1^{(m)} = \binom{d}{m}-\binom{d-1}{m} = \binom{d-1}{m-1} = \beta^{(m)}. 
\end{align*}
\end{remark}

\begin{remark}
It is worth mentioning that the repair-bandwidth proposed in Theorem~\ref{thm:MulRep} is universally and simultaneously achievable. That is, the same determinant code can simultaneously achieve $\beta_e^{(m)}$ for every $e\in\{1,2,\dots, n-d\}$. 
\end{remark}

The next theorem shows that the repair bandwidth for multiple failures can be further reduced in the centralized repair setting \cite{ye2017explicit, zorgui2017centralized}, by a sequential repair mechanism, and exploiting the repair symbols contributed by already repaired failed nodes which can act as helpers. 

\begin{thm}
In an $(n,k=d,d)$ determinant code with (up to a scalar factor) parameters  
$\alpha^{(m)} = \binom{d}{m}$ and $F^{(m)}=m\binom{d+1}{m+1}$,	
any set of $e$ simultaneously failed nodes can be centrally repaired by accessing an arbitrary set of $d$ helper nodes and downloading a total of 
	\begin{align}
 \bar{\beta}_e^{(m)} = \frac{1}{d} \left[m\binom{d+1}{m+1}-m\binom{d-e+1}{m+1}\right]
\end{align}
repair symbols from each helper node. 
\label{thm:MulRep-improved}
\end{thm}

\begin{remark}
It is worth noting that for $e>d-m$ we have $\bar{\beta}_e^{(m)} =\frac{m}{d}\binom{d+1}{m+1} = F^{(m)}/d$ (independent of $e$), and 
\[
\frac{\bar{\beta}_e^{(m)}}{\alpha^{(m)}} = \frac{m(d+1)}{d(m+1)},
\]
which is strictly less than $1$ as shown in Fig.~\ref{fig:MultipleRepair} (for all corner points except the MSR point, $m=d$). The fact that $\bar{\beta}_e^{(m)} = F^{(m)}/d$ implies that the helper nodes contribute just enough number of repair symbols  to be able to recover the entire file, without sending any redundant data. It is clear that this repair-bandwidth is optimum for $e\geq d$, since such a set of $e$ failed nodes should be able to recover the entire file after being repaired. 
\end{remark}

This theorem is built on the result of Theorem~\ref{thm:MulRep}, by exploiting the repair data can be exchanged among the failed nodes. Note that in the centralized repair setting, the information exchanged among the failed nodes at the repair center are not counted against the repair bandwidth. We prove this theorem in Section~\ref{sec:MulRep-improved}.

\section{Construction of $(n,k=d,d)$ determinant codes}
\label{sec:construction}
The code construction described in this section is identical to that of \cite{elyasi2016determinant}, except the repair process which is different and simpler. However, for the sake of completeness, we start with the details of code construction.
\subsection{Code Construction}
For a distributed storage system with parameters $(n,k=d,d)$ and corresponding to a \emph{mode} $m\in\set{1,2,\dots, d}$, our construction provides an exact-repair regenerating code with per-node storage capacity  $\alpha^{(m)} = \binom{d}{m}$ and per-node repair-bandwidth  $\beta ^{(m)}= \binom{d-1}{m-1}$. This code can store up to $F^{(m)}= m\binom{d+1}{m+1}$ symbols.

We represent the coded symbols in a matrix $\cC_{n \times \alpha}$, in which the $i$-th row corresponds to the encoded data to be stored in $i$-th node of DSS. The proposed code is \emph{linear}, i.e., the encoded matrix $\cC$  is obtained by multiplying an encoder matrix $\enc_{n \times d}$ and a message matrix\footnote{The number of entries in this matrix is more than $F$, the size of the file to be coded. Indeed, there are some redundancies among the entries of this matrix as will be explained later.  $\SM_{d \times \alpha}$, whose construction  will be explained later}. All entries of the encoder matrix and the message matrix\footnote{In general elements of the message matrix can be chosen from a Galois field  $\mathbb{F} = \mathsf{GF}(q^s)$ for some prime number $q$ and an integer $s$.} are assumed to be from a finite field $\mathbb{F}$, which has at least $n$ distinct elements. Moreover, all the arithmetic operations are performed with respect to the underlying finite field. The structures of the encoder and message matrices are given below.

\underline{\textbf{Encoder Matrix:}} The matrix $\enc_{n \times d}$ is a fixed matrix which is shared among all the nodes in the system. The main property required for matrix $\enc$ is being Maximum-Distance-Separable (MDS), that is, any  $d\times d$ sub-matrix  of $\enc_{n \times d}$ is full-rank. Examples of MDS matrices include Vandermonde or Cauchy matrices. We can always convert an MDS  matrix to a systematic MDS matrix, by multiplying it by the inverse of its top $d\times d$ sub-matrix (see the example in Section~\ref{sec:example}). We refer to the first $k$ nodes by \emph{systematic} nodes if a systematic MDS matrix is used for encoding.

\underline{\textbf{Message Matrix:}} The message matrix $\SM$ is filled with raw (source) symbols and parity symbols. Recall that $\SM$  is a $d\times \alpha$ matrix, that has $d\alpha = d\binom{d}{m}$ entries, while we wish to store only $F=m\binom{d+1}{m+1}$ source symbols. Hence, there are $d\alpha-F= \binom{d}{m+1}$ redundant entries in $\SM$, which are filled with  parity symbols. More precisely, we divide the set of $F$ data symbols into two groups, namely, $ \sen V$ and $\sen W$, whose elements are indexed by sets as follows
\begin{align*}
\sen V  &= \left\{v_{x,\sen I}:{\sen I} \subseteq \intv d,\size {\sen  I} = m , x\in \sen I \right\},\\
\sen W &=	\left\{w_{x,\sen I}:{\sen I} \subseteq \intv d,\size {\sen  I} =m+1, x\in {\sen I},\ind{\sen I}{x} \leq m\right\}.
\end{align*}	
Note that each element of $\sen V$ is  indexed by a set ${\sen  I} \subseteq \intv d$ of length $|{\sen I}|=m$ and an integer number $x\in {\sen I}$. Hence, $\size{\sen V}=m\binom{d}{m}$. Similarly, symbols in  $\sen W$ are indexed by a pair $(x, \sen I)$, where $\sen I$ is a subset of $\intv{d}$ with $m+1$ entries, and $x$  can take any value in $\sen I$ except the largest one. So, there are $\size {\sen W} = m\binom{d}{m+1}$ symbols in set $\sen W$. Note that $F=|\sen V| + |\sen W|$. 

For the sake of completeness, we define parity symbols indexed by pairs $(x, \sen I)$, where $\size{\sen I}=m+1$ and $\ind{\sen I}{x}=m+1$. Such symbols are constructed such that  parity equations  
\begin{align}
\sum_{y \in \sen I} (-1)^{\ind{\sen I}{y}} w_{y, \sen I} = 0,
\label{eq:parityeq}
\end{align}
hold for any $\sen I \subseteq \intv{d}$ with $|\sen I|=m+1$. In other words, such missing symbols are given by\footnote{Note that for an underlying Galois field $\mathsf{GF}(2^s)$ with characteristic $2$, the parity equation reduces to $\sum_{y \in \sen I} w_{y,\sen I} =0$.} $(-1)^{m+1}w_{\max \sen I, \sen I}=-\sum_{y \in \sen I\setminus \set{\max \sen I}} (-1)^{\ind{\sen I}{y}} w_{y, \sen I} $. 

The rows of matrix $\SM$ are labeled by numbers $1, 2, \dots, d$, and the columns are labeled by subsets $\sen I$ of $\intv{d}$ of size  $|\sen I| = m$. The entries of matrix $\SM$ are given by
\begin{align}
\SM_{x,\sen I}=\left\{
\begin{array}{l l}      
v_{x,\sen I} & \textrm{if $x \in \sen I$}, \\
w_{x,\sen{I}\cup \set {x}} & \textrm{if $x \notin \sen I$}. 
\end{array}\right.
\label{eq:def:S}
\end{align}
It is shown in \cite[Proposition~1]{elyasi2016determinant} that the entire data encoded by this code can be recovered from the content of any $k=d$ nodes. 
Next, we show the exact-repair properties for single and multiple failures.


\subsection{Single Failure Exact Repair}
\label{sec:repair}
The second important property of the proposed code is its ability to exactly repair the content of a failed node using the repair data sent by the helper nodes. Let node $f\in \intv{n}$ fails, and a set of helper nodes $\mathcal{H}\subseteq \{1,2,\dots,n\}\setminus\{f\}$ with $|\mathcal{H}|=d$ wishes to repair node $f$. We first determine the repair data sent from each helper node in order to  repair node $f$. 

\noindent \underline{\textbf{Repair Encoder Matrix at the Helper Nodes:}}
	For a determinant code operating in mode $m$ and a failed node $f$, the repair-encoder matrix $\repMat{f}{m}$ is defined as a $\binom{d}{m} \times \binom{d}{m-1}$ matrix, whose rows are labeled by $m$-element subsets of $\intv{d}$ and columns are labeled by $(m-1)$-element subsets of $\intv{d}$. The entry in  row $\sen I$ and  column $\sen J$  is given by
	\begin{align}
	\repMat{f}{m}_{\sen I, \sen J}=\left\{
	\begin{array}{l l}      
	(-1)^{\ind{\sen I}{x}}\psi_{f,x} & \textrm{if $\sen I \cup \set{x} =  \sen J$},\\
	0& \textrm{otherwise},
	\end{array}\right.
	\label{eq:def:xi}
	\end{align}
where $\psi_{f,x}$ is the entry of the encoder matrix $\enc$ at position $(f,x)$. An example of the $\Xi$ matrix is given in \eqref{eq:ex:Xi} in Section~\ref{sec:example}.

In order to repair node $f$, each helper node $h\in \sen H$  multiplies its content $\enc_h  \cdot \SM$ by the repair-encoder matrix of node $f$ to obtain $\enc_h \cdot \SM \cdot \repMat{f}{m}$, and sends it to node $f$. Note that matrix $\repMat{f}{m}$ has $\binom{d}{m-1}$ columns, and hence the length of the repair data $\enc_h \cdot \SM \cdot \repMat{f}{m}$ is $\binom{d}{m-1}$, which is greater than $\beta=\binom{d-1}{m-1}$. However, the following proposition states that out of $\binom{d}{m-1}$ columns of matrix $\repMat{f}{m}$  at most $\beta^{(m)}=\binom{d-1}{m-1}$ are linearly independent. Thus, the entire vector $\enc_h \cdot \SM \cdot \repMat{f}{m}$ can be sent by communicating at most $\beta$ symbols (corresponding to the linearly independent columns of $\repMat{f}{m}$) to the failed node, and other symbols can be reconstructed using the linear dependencies among the repair symbols.  This is formally stated in the following proposition, which is proved in Section~\ref{sec:nddproofs}.
\begin{prop}
	The rank of matrix $\repMat{f}{m}$ is at most $\beta^{(m)}=\binom{d-1}{m-1}$.
	\label{lm:beta}
\end{prop}

\noindent \underline{\textbf{Decoding at the Failed Node:}}
Upon receiving $d$ repair-data vectors  $\set{\enc_h\cdot \SM \cdot \repMat{f}{m}:h \in \sen H}$, the failed node  stacks them to form a matrix  $\enc[\hlp,:] \cdot \SM \cdot \repMat{f}{m}$, where $\enc[\hlp,:]$ in the sub-matrix of $\enc$ obtained from nodes  $h\in\cH$. This matrix is  full-rank  by the definition of the $\Psi$ matrix.  Multiplying by $\enc[\hlp,:]^{-1}$, the failed node retrieves 
\begin{align}
\repSpc{f}{m}= \SM \cdot \repMat{f}{m}.
\end{align}
This is a $d\times \binom{d}{m-1}$ matrix. These $d \binom{d}{m-1}$ linear combinations of the data symbols span a linear subspace, which we refer to by \emph{repair space of node $f$}. The following proposition shows that all of the missing symbols of node $f$ can be recovered from its repair space.

 \begin{prop}
\label{prop:node:rep} 	
In the $(n,k=d,d)$ proposed  codes with parameters $(\alpha^{(m)},\beta^{(m)},F^{(m)})=\left(\binom{d}{m},\binom{d-1}{m-1},m\binom{d+1}{m+1} \right)$, for every failed node $f\in \intv{n}$ and set of helpers $\sen H \subseteq \intv {n} \setminus\set{f}$ with $\size{H}=d$,  the content of node $f$  can be exactly regenerated by downloading $\beta$ symbols from each of nodes in $\sen H$. More precisely,  the $\sen I$-th entry of the node $f$ can be recovered using
\begin{align}
\left[\enc_{f} \cdot \SM\right]_{\sen I} = \sum_{x\in \sen I} (-1)^{\ind{\sen I}{x}} \left[\repSpc{f}{m}\right]_{x,\sen I \setminus \seq{x}}.
\label{eq:repair}
\end{align}
\end{prop}	
The proof of this proposition is presented in Section~\ref{sec:nddproofs}.

\begin{rem}
Note that for a code defined on the Galois field $\mathsf{GF}(2^s)$ with characteristic $2$, we have $-1=+1$, and hence, all the positive and negative signs disappear. In particular, the parity equation in  \eqref{eq:parityeq} will simply reduce to $\sum_{y \in \sen I} w_{y, \sen I} = 0$, the non-zero entries of the repair encoder matrix in \eqref{eq:def:xi} will be $\psi_{f,x}$, and the repair equation in \eqref{eq:repair} will be replaced by $\left[\enc_{f} \cdot \SM\right]_{\sen I} = \sum_{x\in \sen I} \left[\repSpc{f}{m}\right]_{x,\sen I \setminus \seq{x}}$.
\end{rem}

\subsection{Multiple Failure Exact Repair}
\label{sec:multi-repair}
The repair mechanism proposed for multiple failure scenario is similar to that of the single failure case. We consider a set of failed nodes $\sen E$ with $e=|\sen E|$ failures. Each helper node $h\in \sen H$ sends its repair data to all failed nodes simultaneously. Each failed node $f\in \sen E$ can recover the repair data $\set{\enc_h\cdot \SM \cdot \repMat{f}{m}:h \in \sen H}$, and the repair mechanism is similar to that explained in Proposition~\ref{prop:node:rep}. 

A naive approach is to simply concatenate all the required repair data $\set{\enc_h\cdot \SM \cdot \repMat{f}{m}:f \in \sen E}$ at the helper node $h\in \sen H$ and send it to the failed nodes. More precisely, for a set of failed nodes $\sen E=\{f_1,f_2,\dots, f_e\}$ and a helper node $h\in \sen H$, we define its repair data as $\enc_h\cdot \SM \cdot \repMat{\sen E}{m}$, where 
\begin{align}
		\repFam{\sen E}{m}= \left[ \begin{array}{c|c|c|c} \repMat{f_1}{m} &\repMat{f_2}{m} & \cdots & \repMat{f_e}{m} \end{array} \right]. 
		\label{eq:def:xi-multi}
	\end{align}
This is simply a concatenation of the repair data for individual repair of $f\in \sen E$, and the content of each failed node can be exactly reconstructed according to Proposition~\ref{prop:node:rep}. The repair bandwidth required for naive concatenation scheme is $e\times \beta^{(m)}_1 = e\binom{d-1}{m-1}$. Instead, we show that the bandwidth can be opportunistically utilized by exploiting the intersection between the repair space of the different failed nodes. The following proposition shows that the repair data $\enc_h\cdot \SM \cdot \repMat{\sen E}{m}$ can be delivered to the failed nodes by communicating only $\beta_e^{(m)}$ repair symbols.

\begin{prop}
	Assume that a family of $e$ nodes $\sen E = \set{f_1,f_2,\cdots,f_e}$ are failed. Then the rank  of matrix $\repFam{\sen E}{m}$ defined in \eqref{eq:def:xi-multi}	is at most $\binom{d}{m}-\binom{d-e}{m}$.
\label{prop:multi-beta}
\end{prop}

\section{An Illustrative Example for $(n=8,k=4,d=4)$ codes}
\label{sec:example}
Before presenting the formal proof of the main properties of the proposed code, we show the code construction and the repair mechanism through an example in this section. This example is similar to that of \cite{elyasi2016determinant}, and will be helpful to understand the notation and the details of the code construction, as well as to provide an intuitive justification for its underlying properties.

Let's consider a distributed storage system with parameters $(n,k,d)=(8,4,4)$ and an operating mode  $m=2$. The parameters of the proposed regeneration code for this point of the trade-off are given by

\begin{align}
\left(\!\alpha^{(2)}, \beta^{(2)}, F^{(2)}\!\right)\!=\! \left(\!\binom{4}{2} 
, \binom{4-1}{2-1}
, 2\binom{4+1}{2+1}\!\right)=\left(6,3,20\right).
\end{align} 

We first label and partition the information symbols into two groups, $\sen V$ and $\sen W$, with $|\sen V|=m \binom{d}{m} = 2 \binom{4}{2} = 12$ and $|\sen W| = m \binom{d}{m+1}= 2\binom{4}{3}= 8$. Note that $|\cV| + |\cW|= 20 =F$. 
\begin{align*}
\sen V &=\left\{v_{1,\seq{1,2}},v_{2,\seq{1,2}},v_{1,\seq{1,3}},v_{3,\seq{1,3}},
v_{1,\seq{1,4}},v_{4,\seq{1,4}}, v_{2,\seq{2,3}},  v_{3,\seq{2,3}}, v_{2,\seq{2,4}},v_{4,\seq{2,4}},        
v_{3,\seq{3,4}},v_{4,\seq{3,4}}\right\},\\
\sen W &=\left\{ w_{1,\seq{1,2,3}},w_{2,\seq{1,2,3}},w_{1,\seq{1,2,4}},w_{2,\seq{1,2,4}}, w_{1,\seq{1,3,4}},w_{3,\seq{1,3,4}},w_{2,\seq{2,3,4}},w_{3,\seq{2,3,4}}\right\}.
\end{align*}
Moreover, for each subset $\sen I\subseteq \intv{4}$ with $|\sen I|=m+1=3$, we define parity symbols as
\begin{align}
\begin{split}
\left\{
\begin{array}{l}
\sen I=\set{1,2,3}:w_{3,\seq{1,2,3}}=w_{2,\seq{1,2,3}}-w_{1,\seq{1,2,3}}, \\
\sen I=\set{1,2,4}:w_{4,\seq{1,2,4}}=w_{2,\seq{1,2,4}}-w_{1,\seq{1,2,4}}, \\
\sen I=\set{1,3,4}:w_{4,\seq{1,3,4}}=w_{3,\seq{1,3,4}}-w_{1,\seq{1,3,4}}, \\
\sen I=\set{2,3,4}:w_{4,\seq{2,3,4}}=w_{3,\seq{2,3,4}}-w_{2,\seq{2,3,4}}. 
\end{array}\right.
\end{split}
\label{eq:ex:w-par}
\end{align}

Next, the message matrix $\SM$ will be formed by placing  $v$ and $w$ symbols as specified in \eqref{eq:def:S}. The resulting message matrix is given by
\begin{figure*}[!h]
\centering
\includegraphics[width=0.8\textwidth]{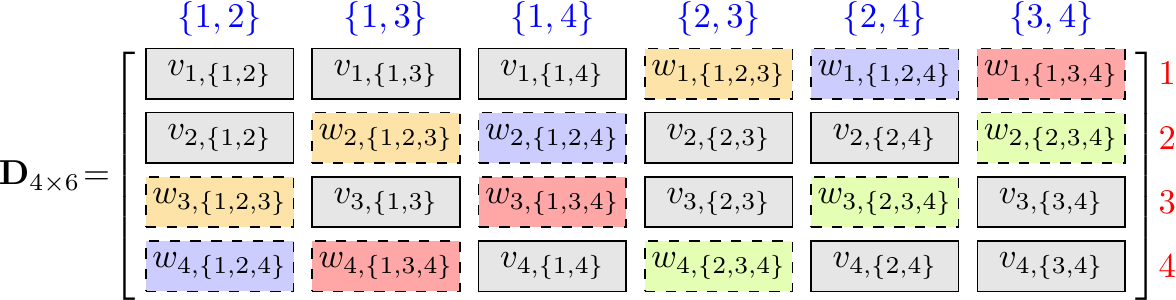}
\end{figure*}

The next step for encoding the data is multiplying $\SM$ by an encoder matrix $\enc$. To this end, we choose a finite field $\mathbb{F}_{13}$ (with at least $n=8$ distinct non-zero entries), and pick an $8 \times 4$ Vandermonde matrix generated by $i=1,2,3,4,5,6,7,8$. We convert this matrix to a systematic MDS matrix by multiplying it from the right by the inverse of its top $4\times 4$ matrix.  That is, 
\begin{align*}
\enc_{8 \times 4} &= \left[
\begin{array}{c}
\enc_1 \\
\enc_2 \\
\enc _3 \\
\enc _4 \\
\enc _5 \\
\enc _6 \\
\enc_7 \\
\enc_8 \\
\end{array}
\right ]
=
\left[
\begin{array}{cccc}
1^0 & 1^1 & 1^2 & 1^3 \\
2^0 & 2^1 & 2^2 & 2^3 \\
3^0 & 3^1 & 3^2 & 3^3 \\
4^0 & 4^1 & 4^2 & 4^3 \\
5^0 & 5^1 & 5^2 & 5^3 \\
6^0 & 6^1 & 6^2 & 6^3 \\
7^0 & 7^1 & 7^2 & 7^3 \\
8^0 & 8^1 & 8^2 & 8^3 
\end{array}
\right]
\cdot
\left[
\begin{array}{cccc}
1^0 & 1^1 & 1^2 & 1^3 \\
2^0 & 2^1 & 2^2 & 2^3 \\
3^0 & 3^1 & 3^2 & 3^3 \\
4^0 & 4^1 & 4^2 & 4^3
\end{array}
\right]^{-1}
=
\left[
\begin{array}{ccccc}
1 & 0 & 0 & 0\\
0 & 1 & 0 & 0\\
0& 0 & 1 & 0\\
0 & 0 & 0 & 1\\
12 & 4 & 7 & 4\\
9 & 2 & 6 & 10\\
3 & 10 & 7 & 7\\
6 & 5 & 7 & 9\\
\end{array}
\right] \quad (\text{mod\;} 13).
\end{align*}

Note that every $k=4$ rows of matrix $\enc$ are linearly independent, and form an invertible matrix. 
Then the content of node $i$ is formed by row $i$ in the matrix product $\enc \cdot \SM$, which we denote by $\enc_i \cdot \SM$. 

Data recovery from the content of any $k=4$ node is immediately implied by the MDS property of the encoder matrix. For further details, we refer to Section~IV in \cite{elyasi2016determinant}. Next, we describe the repair process for single and multiple failures.

\subsection{Single Failure Repair}
\label{sec:ex:1rep}

First, suppose that a non-systematic node $f$ fails, and we wish to repair it by the help of the systematic nodes $\sen H=\{1,2,3,4\}$, by downloading $\beta = 3$  from each. The content of node $f$ is given by $\enc_f\cdot \SM$, which includes $\alpha=6$ symbols. Note that the content of this node is a row vector whose elements has the same labeling as the columns of $\SM$, i.e all $m=2$ elements subsets of $\intv{d}=\set{1,2,3,4}$. The symbols of this node are given by:
\begin{align}
\begin{split}
\left\{
\begin{array}{lllll}
\textcolor{blue}{\sen I = \set {1,2}:}  &\psi _{f,1} v_{1,\{1,2\}} &\hspace{-3mm}+ \psi _{f,2} v_{2,\{1,2\}}&\hspace{-3mm}+ \psi _{f,3} w_{3,\{1,2,3\}} &\hspace{-3mm}+ \psi _{f,4} w_{4,\{1,2,4\}},\\
\textcolor{blue}{\sen I = \set {1,3}:}  &\psi _{f,1} v_{1,\{1,3\}} &\hspace{-3mm}+ \psi _{f,2} w_{2,\{1,2,3\}} &\hspace{-3mm}+ \psi _{f,3} v_{3,\{1,3\}} &\hspace{-3mm}+ \psi _{f,4} w_{4,\{1,3,4\}},\\
\textcolor{blue}{\sen I = \set {1,4}:}  &\psi _{f,1} v_{1,\{1,4\}} &\hspace{-3mm}+ \psi _{f,2} w_{2,\{1,2,4\}}  &\hspace{-3mm}+  \psi _{f,3} w_{3,\{1,3,4\}} &\hspace{-3mm}+  \psi _{f,4} v_{4,\{1,4\}},\\
\textcolor{blue}{\sen I = \set {2,3}:}  &\psi _{f,1} w_{1,\{1,2,3\}} &\hspace{-3mm}+  \psi _{f,2} v_{2,\{2,3\}} &\hspace{-3mm}+  \psi _{f,3} v_{3,\{2,3\}} &\hspace{-3mm}+ \psi _{f,4} w_{4,\{2,3,4\}}, \\
\textcolor{blue}{\sen I = \set {2,4}:}  &\psi_{f,1} w_{1,\{1,2,4\}} &\hspace{-3mm}+ \psi _{f,2} v_{2,\{2,4\}} &\hspace{-3mm}+ \psi _{f,3} w_{3,\{2,3,4\}} &\hspace{-3mm}+  \psi _{f,4} v_{4,\{2,4\}},\\
\textcolor{blue}{\sen I = \set {3,4}:} &\psi _{f,1} w_{1,\{1,3,4\}} &\hspace{-3mm}+ \psi _{f,2} w_{2,\{2,3,4\}} &\hspace{-3mm}+ \psi _{f,3} v_{3,\{3,4\}} &\hspace{-3mm}+ \psi _{f,4} v_{4,\{3,4\}}.
\end{array}
\right.
\end{split}
\label{eq:ex:node-3}
\end{align}

In the repair procedure using the systematic nodes as helpers, every symbol will be repaired by $m$ nodes. Recall that $d$ helper nodes contribute in the repair process by sending $\beta=\binom{d-1}{m-1}$ symbols each, in order to repair $\alpha=\binom{d}{m}$ missing symbols. Hence, the number of repair equations per missing symbol is $d\beta/\alpha = m$, which matches with the proposed repair mechanism. 

The $m=2$ helpers for each missing encoded symbol are those who have a copy of the corresponding $v$-symbols, e.g., for the symbol indexed by $\sen I = \set {1,2}$ which has $v_{1,\{1,2\}}$ and $v_{2,\{1,2\}}$, the contributing helpers are nodes $1$ (who has a copy of $v_{1,\{1,2\}}$) and node $2$ (who stores a copy of $v_{2,\{1,2\}}$). To this end, node $1$ can send $\psi _{f,1} v_{1,\{1,2\}} $ which node $2$ sends $\psi _{f,2} v_{2,\{1,2\}} $ to perform the repair. 

It can be seen that the $\set{1,2}$-th missing symbols has also two other terms depending on  $w_{3,\{1,2,3\}}$ and $w_{4,\{1,2,4\}}$, which are stored at nodes $3$ and $4$, respectively. A naive repair mechanism requires these two nodes  also to contribute in this repair procedure, which yields in a full data recovery in order to repair a failed node. Alternatively, we can reconstruct these $w$-symbols using the parity equations, and the content of the the first two helper nodes. Recall from \eqref{eq:parityeq} that
\begin{align*}
w_{3,\{1,2,3\}} &= -w_{1,\{1,2,3\}} +w_{2,\{1,2,3\}},\\
w_{4,\{1,2,4\}} &= - w_{1,\{1,2,4\}} +w_{2,\{1,2,4\}},
\end{align*}
where $w_{1,\{1,2,3\}}$ and $w_{1,\{1,2,4\}}$ are stored in node $1$ and $w_{2,\{1,2,3\}}$ and $w_{2,\{1,2,4\}}$ are stored in node $2$. Hence, the content of nodes $1$ and $2$ are sufficient to reconstruct the $\set{1,2}$-th symbol at the failed node $f$. To this end, node $1$ computes $A= \psi _{f,1} v_{1,\{1,2\}} - \psi _{f,3} w_{1,\{1,2,3\}}- \psi _{f,4} w_{1,\{1,2,4\}}$ (a linear combination of its first, forth, and fifth entries), and sends $A$ to $f$. Similarly, node $2$ sends  $B= \psi _{f,2} v_{2,\{1,2\}}+ \psi _{f,3} w_{2,\{1,2,3\}}+ \psi _{f,4} w_{2,\{1,2,4\}}$ (a linear combination  of its first, second, and third coded symbols).  Upon receiving these symbols, the $\set{1,2}$-th missing symbols of node $f$ can be recovered from 
\begin{align}
A+B&=\left(\psi _{f,1} v_{1,\{1,2\}} - \psi _{f,3} w_{1,\{1,2,3\}}- \psi _{f,4} w_{1,\{1,2,4\}}\right)+\left(\psi _{f,2} v_{2,\{1,2\}}+ \psi _{f,3} w_{2,\{1,2,3\}}+ \psi _{f,4} w_{2,\{1,2,4\}}\right) \nonumber\\
&= \psi _{f,1} v_{1,\{1,2\}}+ \psi _{f,2} v_{2,\{1,2\}}+ \psi _{f,3} \left(w_{2,\{1,2,3\}}-w_{1,\{1,2,3\}}\right) +\psi _{f,4} \left(w_{2,\{1,2,4\}}-w_{1,\{1,2,4\}}  \right)\nonumber\\
&= \psi _{f,1} v_{1,\{1,2\}} +\psi _{f,2} v_{2,\{1,2\}} + \psi _{f,3} w_{3,\{1,2,3\}} +\psi _{f,4} w_{4,\{1,2,4\}}. \label{eq:ex:rep:parity}
\end{align}
In general,  $v$ symbols are repaired directly by communicating an identical copy of them, while $w$ symbols are repaired indirectly, using their parity equations. This is the general rule that we use for repair of all other missing symbols of node $f$. It can be seen that each helper node participates in the repair of $\beta=3$ missing symbols, by sending one repair symbol for each. For instance, node $1$ contributes in the repair of symbols indexed by $\set{1,2}$, $\set{1,3}$, and $\set{1,4}$. The repair equation sent by node $1$ for each these repair scenarios are listed below: 
\begin{align}
\mathsf{repair\ symbols\ sent\ by\ node\;}1:\begin{split}
\left\{
\begin{array}{l}
\textcolor{blue}{\sen I = \set {1,2}:}  \psi _{f,1} v_{1,\{1,2\}}- \psi _{f,3} w_{1,\{1,2,3\}} - \psi _{f,4} w_{1,\{1,2,4\}}, \\
\textcolor{blue}{\sen I = \set {1,3}:}  \psi _{f,1} v_{1,\{1,3\}} + \psi _{f,2} w_{1,\{1,2,3\}}- \psi _{f,4} w_{1,\{1,3,4\}}, \\
\textcolor{blue}{\sen I = \set {1,4}:}  \psi _{f,1} v_{1,\{1,4\}} + \psi _{f,2} w_{1,\{1,2,4\}} + \psi _{f,3} w_{1,\{1,3,4\}}. 
\end{array}
\right.
\end{split}
\label{eq:ex:rep:node:1}
\end{align}
Similarly, the repair symbols sent from helper nodes $2$, $3$, and $4$ are given by
\begin{eqnarray}
&\mathsf{repair\ symbols\ sent\ by\ node\;}2:&\begin{split}
\left\{
\begin{array}{l}
\textcolor{blue}{\sen I = \set {1,2}:}  \psi _{f,2} v_{2,\{1,2\}}+ \psi _{f,3} w_{2,\{1,2,3\}}+  \psi _{f,4} w_{2,\{1,2,4\}},\\
\textcolor{blue}{\sen I = \set {2,3}:}  \psi _{f,1} w_{2,\{1,2,3\}} + \psi _{f,2} v_{2,\{2,3\}}-  \psi _{f,4} w_{2,\{2,3,4\}}, \\
\textcolor{blue}{\sen I = \set {2,4}:} \psi _{f,1} w_{2,\{1,2,4\}}  +  \psi _{f,2} v_{2,\{2,4\}} +  \psi _{f,3} w_{2,\{2,3,4\}},
\end{array}
\right.
\end{split}\\
&\mathsf{repair\ symbols\ sent\ by\ node\;}3:&\begin{split}
\left\{
\begin{array}{l}
\textcolor{blue}{\sen I = \set {1,3}:}  \psi _{f,2} w_{3,\{1,2,3\}} + \psi _{f,3} v_{3,\{1,3\}}+  \psi _{f,4} w_{3,\{1,3,4\}}, \\
\textcolor{blue}{\sen I = \set {2,3}:} - \psi _{f,1} w_{3,\{1,2,3\}} + \psi _{f,3} v_{3,\{2,3\}} +  \psi _{f,4} w_{3,\{2,3,4\}},\\
\textcolor{blue}{\sen I = \set {3,4}:}  \psi _{f,1} w_{3,\{1,3,4\}} +  \psi _{f,2} w_{3,\{2,3,4\}} +  \psi _{f,3} v_{3,\{3,4\}},
\end{array}
\right.
\end{split}\\
&\mathsf{repair\ symbols\ sent\ by\ node\;}4:&\begin{split}
\left\{
\begin{array}{l}
\textcolor{blue}{\sen I = \set {1,4}:}  \psi _{f,2} w_{4,\{1,2,4\}} +  \psi _{f,3} w_{4,\{1,3,4\}} + \psi _{f,4}v_{4,\{1,4\}}, \\
\textcolor{blue}{\sen I = \set {2,4}:} - \psi _{f,1} w_{4,\{1,2,4\}} +  \psi _{f,3} w_{4,\{2,3,4\}} +  \psi _{f,4} v_{4,\{2,4\}}, \\
\textcolor{blue}{\sen I = \set {3,4}:} - \psi _{f,1} w_{4,\{1,3,4\}} -  \psi _{f,2} w_{4,\{2,3,4\}} +  \psi _{f,4} v_{4,\{3,4\}}.
\end{array}
\right.
\end{split}
\label{eq:ex:rep:node:4}
\end{eqnarray}

The repair symbols of helper node $h\in \{1,2,3,4\}$ in \eqref{eq:ex:rep:node:1}-\eqref{eq:ex:rep:node:4} could be obtain from  $\enc_h \cdot \SM \cdot \repMat{f}{2}$, which is the content of the helper nodes (i.e., $\enc_h \cdot \SM$) times the repair encoder matrix  for $m=2$ (i.e., $\repMat{f}{2}$) defined in~\eqref{eq:def:xi}: 
\begin{align}
\begin{split}
\begin{tikzpicture}
\matrix (Xi4) [matrix of math nodes,%
ampersand replacement=\&,
column sep=0pt,
row sep = 3pt,
column 1/.style={colIndColor,column sep=0pt},
nodes={anchor= center,minimum height=15pt}
]
{
    \phantom{{ }}\&\phantom{{ }}\&|[color=repIndColor]|\seq{1} \& |[color=repIndColor]|\seq{2} \& |[color=repIndColor]|\seq{3} \& |[color=repIndColor]|\seq{4} \\
    \seq{1,2}\& \phantom{{ }} \& \psi _{f,2} \& -\psi _{f,1} \& 0 \& 0 \\
    \seq{1,3}\& \phantom{{ }} \&\psi _{f,3} \& 0 \& -\psi _{f,1} \& 0 \\
    \seq{1,4}\& \phantom{{ }} \&\psi _{f,4} \& 0 \& 0 \& -\psi _{f,1} \\
    \seq{2,3}\& \phantom{{ }} \&0 \& \psi _{f,3} \& -\psi _{f,2} \& 0 \\
    \seq{2,4}\& \phantom{{ }} \&0 \& \psi _{f,4} \& 0 \& -\psi _{f,2} \\
    \seq{3,4}\& \phantom{{ }} \&0 \& 0 \& \psi _{f,4} \& -\psi _{f,3} \\
};
\node[left delimiter=[, xshift=10pt, yshift=1pt, fit={(Xi4-2-2.north west)(Xi4-7-2.south west)}] (brleft){};
\node[right delimiter={]},fit={($ (Xi4-2-6.north east) + (0pt,0pt)$)($(Xi4-7-6.south east) + (0pt,0pt)$)}] (brright){};  
\node[left = 0pt of Xi4-4-1]{$\repMat{f}{2}=$};
\end{tikzpicture}
\end{split}.
\label{eq:ex:Xi}
\end{align}   
Note that, even though this matrix has $\binom{4}{2-1}=4$ columns, and hence, $\enc_h \cdot \SM \cdot \repMat{f}{2}$ is a vector of length  $4$, it suffices to communicate only\footnote{Indeed the entire repair process can be expressed in terms of the  $\beta=3$ repair symbols sent by the helper nodes. However, the recovery equations for the missing symbols are mathematically more symmetric and compact  if we allow  the fourth symbol to appear in the repair equations.}  $\beta=3$ symbols from the helper node to the failed node and the fourth symbol can be reconstructed from the other $3$ symbols at the failed node. This is due to the fact that the rank of matrix $\repMat{f}{2}$ equals to $\beta =3$. More precisely, a non-zero linear combination of the columns of $\repMat{f}{2}$ is zero, that is,
\begin{align}
\repMat{f}{2} \cdot \begin{bmatrix}
\psi_{f,1} & \psi_{f,2} & \psi_{f,3} & \psi_{f,4} 
\end{bmatrix}^{\intercal}=\mathbf{0}.
\label{eq:zero-lin-comb}
\end{align}
Therefore, (if $\psi_{f,4}\neq 0$) the helper node $h$ only sends the first $\beta=3$ symbols of the vector  $\enc_h \cdot \SM \cdot \repMat{f}{2}$, namely, $[\enc_h \cdot \SM \cdot \repMat{f}{2}]_1$, $[\enc_h \cdot \SM \cdot \repMat{f}{2}]_2$, and $[\enc_h \cdot \SM \cdot \repMat{f}{2}]_3$, and the forth symbol $[\enc_h \cdot \SM \cdot \repMat{f}{2}]_4$ can be appended to it at node $f$ from 
\begin{align*}
[\enc_h \cdot \SM \cdot \repMat{f}{2}]_4 = \psi_{f,4}^{-1} \cdot \sum_{i=1}^3 \psi_{f,i} \cdot [\enc_h \cdot \SM \cdot \repMat{f}{2}]_i.
\end{align*}
Upon receiving the repair data from $d=4$ helper nodes $\{1,2,3,4\}$, namely $\{\enc_1 \cdot \SM \cdot \repMat{f}{2}, \enc_2\cdot \SM \cdot\repMat{f}{2}, \enc_3 \cdot\SM \cdot\repMat{f}{2}, \enc_4 \cdot\SM \cdot\repMat{f}{2}\}$, the failed can stack them to obtain a matrix 
\begin{align}
\begin{bmatrix}
\enc_1 \cdot \SM \cdot \repMat{f}{2} \\
\enc_2 \cdot \SM \cdot \repMat{f}{2}\\
\enc_3 \cdot \SM \cdot \repMat{f}{2}\\
\enc_4 \cdot \SM \cdot \repMat{f}{2}
\end{bmatrix} = \enc [\{1,2,3,4\},:]\cdot \SM \cdot \repMat{f}{2} = \SM \cdot \repMat{f}{2},
\label{eq:rep-subspace-systematic}
\end{align}
where the last identity is due to the fact that $\enc [\{1,2,3,4\},:]= \mathbf{I}$ is the identity matrix. We refer to this matrix by the \emph{repair space matrix} of node $f$, and denote it by $\repSpc{f}{2} = \SM \cdot \repMat{f}{2}$, as  presented at the top of the next page.

\begin{figure*}[!t]
	\begin{footnotesize}
		\begin{align*}
		\begin{tikzpicture}[scale=1, every node/.style={scale=1}]
		\matrix (RMAT) [matrix of math nodes,%
		ampersand replacement= \&,
		column sep = 10pt,
		column 1/.style={column sep=0,text width=5},
		nodes={anchor= center}
		]
		{ 
			\phantom{{ }} \& |[color=repIndColor]| \set{1} \& |[color=repIndColor]| \set{2}\&\\
    		|[color=rowIndColor]|1 \& v_{1,\{1,2\}} \psi _{f,2}+v_{1,\{1,3\}} \psi _{f,3}+v_{1,\{1,4\}} \psi _{f,4} \& -v_{1,\{1,2\}} \psi _{f,1}+w_{1,\{1,2,3\}} \psi _{f,3}+w_{1,\{1,2,4\}} \psi _{f,4}\& \cdots \\
            |[color=rowIndColor]|2 \& v_{2,\{1,2\}} \psi _{f,2}+w_{2,\{1,2,3\}} \psi _{f,3}+w_{2,\{1,2,4\}} \psi _{f,4} \& -v_{2,\{1,2\}} \psi _{f,1}+v_{2,\{2,3\}} \psi _{f,3}+v_{2,\{2,4\}} \psi _{f,4}\& \cdots \\
            |[color=rowIndColor]|3 \& w_{3,\{1,2,3\}} \psi _{f,2}+v_{3,\{1,3\}} \psi _{f,3}+w_{3,\{1,3,4\}} \psi _{f,4} \& -w_{3,\{1,2,3\}} \psi _{f,1}+v_{3,\{2,3\}} \psi _{f,3}+w_{3,\{2,3,4\}} \psi _{f,4}\& \cdots \\
            |[color=rowIndColor]| 4 \& w_{4,\{1,2,4\}} \psi _{f,2}+w_{4,\{1,3,4\}} \psi _{f,3}+v_{4,\{1,4\}} \psi _{f,4} \& -w_{4,\{1,2,4\}} \psi _{f,1}+w_{4,\{2,3,4\}} \psi _{f,3}+v_{4,\{2,4\}} \psi _{f,4}\& \cdots \\[3mm]        
			\phantom{{ }} \& |[color=repIndColor]|\set {3} \& |[color=repIndColor]| \set {4}\&\\
			\phantom{{ }} \& -v_{1,\{1,3\}} \psi _{f,1}-w_{1,\{1,2,3\}} \psi _{f,2}+w_{1,\{1,3,4\}} \psi _{f,4} \& -v_{1,\{1,4\}} \psi _{f,1}-w_{1,\{1,2,4\}} \psi _{f,2}-w_{1,\{1,3,4\}} \psi _{f,3} \\
			\phantom{{ }} \& -w_{2,\{1,2,3\}} \psi _{f,1}-v_{2,\{2,3\}} \psi _{f,2}+w_{2,\{2,3,4\}} \psi _{f,4} \& -w_{2,\{1,2,4\}} \psi _{f,1}-v_{2,\{2,4\}} \psi _{f,2}-w_{2,\{2,3,4\}} \psi _{f,3} \\
			\phantom{{ }} \& -v_{3,\{1,3\}} \psi _{f,1}-v_{3,\{2,3\}} \psi _{f,2}+v_{3,\{3,4\}} \psi _{f,4} \& -w_{3,\{1,3,4\}} \psi _{f,1}-w_{3,\{2,3,4\}} \psi _{f,2}-v_{3,\{3,4\}} \psi _{f,3} \\
			\phantom{{ }} \& -w_{4,\{1,3,4\}} \psi _{f,1}-w_{4,\{2,3,4\}} \psi _{f,2}+v_{4,\{3,4\}} \psi _{f,4} \& -v_{4,\{1,4\}} \psi _{f,1}-v_{4,\{2,4\}} \psi _{f,2}-v_{4,\{3,4\}} \psi _{f,3} \\
		}; 
		\node[left delimiter=[, xshift=5pt, yshift=0pt, fit={(RMAT-2-2)(RMAT-3-2)(RMAT-4-2)(RMAT-5-2)}] (brleft){};
		\node[right delimiter={]}, xshift=0pt, yshift=0pt,fit={(RMAT-7-3)(RMAT-8-3)(RMAT-9-3)(RMAT-10-3)}] (brright){};  
		\node[left = 20pt of brleft]{$\repSpc{f}{2}=$};
		\node[right delimiter={|},fit=(RMAT-2-2)(RMAT-3-2)(RMAT-5-2),xshift=3pt,yshift=0pt,scale=1] {};
		\node[right delimiter={|},fit=(RMAT-7-2)(RMAT-9-2)(RMAT-10-2),xshift=0pt,yshift=0pt,scale=1] {};
		\end{tikzpicture}
		\end{align*}
	\end{footnotesize}
	\vspace{-7mm}
	\hrule
	\vspace{-7mm}
\end{figure*}

Using the entries of matrix $\repSpc{f}{2}$, we can reconstruct the missing coded symbols of the failed node, as formulated in~\eqref{eq:repair}.  For the sake of illustration, let us focus on the symbol at the position $\sen I = \set{2,4}$ of node $f$. Recall that rows of matrix $\repSpc{f}{2}$ are indexed by numbers in $\intv{d}=\set{1,2,3,4}$ and its columns are indexed by subsets of size $m-1=1$ of $\intv{d}$. The $\sen I$-th symbol of node $f$ can be found from a linear combination (with $+1$ and $-1$ coefficients) of entries of $\repSpc{f}{2}$ positioned at row $x$ and column $\sen I\setminus \{x\}$ for all $x\in \sen I$. The coefficients used in this linear combination is given by order of number $x$ in set $\sen I$, e.g., $x=2$ is the first (smallest) element of $\sen I=\set{2,4}$, hence $\ind{\set{2,4}}{2}=1$, and the corresponding coefficient will be $(-1)^{\ind{\set{2,4}}{2}}= (-1)^1=-1$. Putting things together, we obtain
	\begin{align}
	\sum_{x \in \set{2,4}} &(-1)^{\ind{\set{2,4}}{x}}\repSpc{f}{2}_{x,\set{2,4} \setminus \set{x}}=-\repSpc{f}{2}_{2,\set{4}}+\repSpc{f}{m}_{4,\set{2}}\nonumber\\ 
     &=	-\left[-\psi _{f,2} v_{2,\{2,4\}}-\psi _{f,1} w_{2,\{1,2,4\}}-\psi _{f,3} w_{2,\{2,3,4\}}\right] 
     +\left[\psi _{f,4} v_{4,\{2,4\}}-\psi _{f,1} w_{4,\{1,2,4\}}+\psi _{f,3} w_{4,\{2,3,4\}}\right] \nonumber\\
	&=\psi _{f,2} v_{2,\{2,4\}}+\psi_{f,4} v_{4,\{2,4\}}+\psi _{f,1}\left( w_{2,\{1,2,4\}}-w_{4,\{1,2,4\}} \right)
	+\psi_{f,3}\left( w_{2,\{2,3,4\}}+ w_{4,\{2,3,4\}}\right)\nonumber \\
	&=\psi _{f,2} v_{2,\{2,4\}}+\psi_{f,4} v_{4,\{2,4\}}+\psi _{f,1}w_{1,\{1,2,4\}}+\psi_{f,3}w_{2,\{2,3,4\}}\label{eq:R:Example}\\
	&=[\enc_3 \cdot \SM]_{\seq{2,4}},\nonumber
	\end{align}
where in~\eqref{eq:R:Example} we used the parity equations defined in~\eqref{eq:ex:w-par}. A general repair scenario with an arbitrary (not necessarily systematic) set of helper nodes $\cH$ with $|\cH|=d=4$ is very similar to that from the systematic nodes, explained above. 

Each helper node $h\in \cH$ computes its repair data by multiplying its content by the repair encoder matrix of failed node $f$, and sends it to the failed node. The failed node collects  $\left\{\enc_h \cdot \SM \cdot \repMat{f}{2}: h\in \cH \right\}$ and stacks them to form the matrix $\enc[\hlp,:] \cdot \SM \cdot \repMat{f}{2}$, where $\enc[\hlp,:]$ is the sub-matrix of $\enc$ obtained by stacking rows indexed by $\{h:h\in \cH\}$.  The main difference compared to the systematic helpers case is that unlike in \eqref{eq:rep-subspace-systematic}, $\enc[\hlp,:]$ is not an identity matrix in general. However, since $\enc[\hlp,:]$  is an invertible matrix, we can compute $\repSpc{f}{2}=\SM \cdot\repMat{f}{2}$ as
\begin{align*}
\repSpc{f}{2} =\enc[\hlp,:]^{-1} \cdot \left(\enc[\hlp,:] \cdot \SM \cdot \repMat{f}{2} \right)=\SM \cdot \repMat{f}{2}. 
\end{align*}  
Once $\repSpc{f}{2}$ is computed at node $f$, the rest of the process is identical the repair from systematic helper nodes.

\subsection{Multiple Failure Repair}
\label{sec:ex:mulrep}
Now, assume two non-systematic nodes in $\sen E=\{f_1,f_2\}$ are simultaneously failed, and our goal is to reconstruct the missing data on $f_1$ and $f_2$ using the systematic nodes, i.e., the helper set is $\sen H=\{1,2,3,4\}$. 
A naive approach is to repeat the repair scenario discussed in Section ~\ref{sec:ex:1rep} for $f_1$ and $f_2$. Such a separation-based scheme requires downloading  $2\beta=6$ (coded) symbols from each helper node. Alternatively, we show that the repair of nodes $f_1$ and $f_2$ can be performed by downloading only $\beta_2=5$ symbols from each helper. 

We start with $\repMat{\{f_1,f_2\}}{2}$ which is basically the concatenation of $\repMat{f_1}{2}$ and $\repMat{f_2}{2}$:
\begin{align*}
	\begin{tikzpicture}
	\matrix (Xi4) [matrix of math nodes,%
	ampersand replacement=\&,
	column sep=0pt,
	row sep = 3pt,
	column 1/.style={colIndColor,column sep=0pt},
    column 6/.style={column sep=10pt},
	nodes={anchor= center,minimum height=15pt}
	]
	{
		\phantom{{ }}\&\phantom{{ }}\&|[color=repIndColor]|\seq{1} \& |[color=repIndColor]|\seq{2} \& |[color=repIndColor]|\seq{3} \& |[color=repIndColor]|\seq{4} \&|[color=repIndColor]|\seq{1} \& |[color=repIndColor]|\seq{2} \& |[color=repIndColor]|\seq{3} \& |[color=repIndColor]|\seq{4} \\
		\seq{1,2}\& \phantom{{ }} \& \psi _{f_1,2} \& -\psi _{f_1,1} \& 0 \& 0 \& \psi _{f_2,2} \& -\psi _{f_2,1} \& 0 \& 0 \\
		\seq{1,3}\& \phantom{{ }} \&\psi _{f_1,3} \& 0 \& -\psi _{f_1,1} \& 0 \&\psi _{f_2,3} \& 0 \& -\psi _{f_2,1} \& 0\\
		\seq{1,4}\& \phantom{{ }} \&\psi _{f_1,4} \& 0 \& 0 \& -\psi _{f_1,1} \&\psi _{f_2,4} \& 0 \& 0 \& -\psi _{f_2,1} \\
		\seq{2,3}\& \phantom{{ }} \&0 \& \psi _{f_1,3} \& -\psi _{f_1,2} \& 0\&0 \& \psi _{f_2,3} \& -\psi _{f_2,2} \& 0 \\
		\seq{2,4}\& \phantom{{ }} \&0 \& \psi _{f_1,4} \& 0 \& -\psi _{f_1,2} \&0 \& \psi _{f_2,4} \& 0 \& -\psi _{f_2,2} \\
		\seq{3,4}\& \phantom{{ }} \&0 \& 0 \& \psi _{f_1,4} \& -\psi _{f_1,3} \&0 \& 0 \& \psi _{f_2,4} \& -\psi _{f_2,3}\\
	};
	\node[left delimiter=[, xshift=10pt, yshift=1pt, fit={(Xi4-2-2.north west)(Xi4-7-2.south west)}] (brleft){};
	\node[right delimiter={]},fit={($ (Xi4-2-10.north east) + (0pt,0pt)$)($(Xi4-7-10.south east) + (0pt,0pt)$)}] (brright){};  
	\node[left = 0pt of Xi4-4-1]{$\repMat{\cE}{2}=\left[\begin{array}{c|c}\repMat{f_1}{2}&\repMat{f_2}{2} \end{array}\right]=$};
	\draw[dashed] ($ (Xi4-2-6.north east) + (17pt,0pt)$)--($ (Xi4-7-6.south east) + (5pt,0pt)$);
	\end{tikzpicture}
\end{align*}
This is a $6\times 8$ matrix. However, the claim of Proposition~\ref{prop:multi-beta} implies the rank of this matrix is at most $5$. To show this claim, we define the non-zero vector 
\begin{align*}
\begin{tikzpicture}
\matrix (Y2) [matrix of math nodes,%
ampersand replacement=\&,
column sep=-4pt,
row sep = 3pt,
column 1/.style={nulIndColor,column sep=-6pt},
nodes={anchor= center,minimum height=10pt}
]
{
	\phantom{{ }}\&\phantom{{ }}\&|[color=colIndColor]|\seq{1,2} \&|[color=colIndColor]|\seq{1,3} \&|[color=colIndColor]|\seq{1,4}\&|[color=colIndColor]|\seq{2,3}\&|[color=colIndColor]|\seq{2,4}\&|[color=colIndColor]|\seq{3,4}\\
	\seq{3,4}\& \phantom{{ }}  \&	-\begin{vmatrix} \psi_{f_1,3} &  \hspace{-2mm}\psi_{f_1,4} \\  \psi_{f_2,3} &  \hspace{-2mm}\psi_{f_2,4}  \end{vmatrix},\&\begin{vmatrix} \psi_{f_1,2} &  \hspace{-2mm}\psi_{f_1,4} \\  \psi_{f_2,2} &  \hspace{-2mm}\psi_{f_2,4}  \end{vmatrix},\& -\begin{vmatrix} \psi_{f_1,2} &  \hspace{-2mm}\psi_{f_1,3} \\  \psi_{f_2,2} &  \hspace{-2mm}\psi_{f_2,3}  \end{vmatrix} ,\& -\begin{vmatrix} \psi_{f_1,1} &  \hspace{-2mm}\psi_{f_1,4} \\  \psi_{f_2,1} &  \hspace{-2mm}\psi_{f_2,4}  \end{vmatrix} ,\& \begin{vmatrix} \psi_{f_1,1} &  \hspace{-2mm}\psi_{f_1,3} \\  \psi_{f_2,1} &\hspace{-2mm}  \psi_{f_2,3}  \end{vmatrix},\& -\begin{vmatrix} \psi_{f_1,1} &\hspace{-2mm}  \psi_{f_1,2} \\  \psi_{f_2,1} &\hspace{-2mm}  \psi_{f_2,2}  \end{vmatrix}\\
};
\node[left delimiter=[, xshift=10pt, yshift=1pt, fit={(Y2-2-3.north west)(Y2-2-3.south west)}] (brleft){};
\node[right delimiter={]},fit={($ (Y2-2-6.north east) + (-5pt,0pt)$)($(Y2-2-8.south east) + (-5pt,0pt)$)}] (brright){};  
\node (eq) [left = -5pt of Y2-2-1,yshift=1mm]{$\NSpc{\cE}{2}=$};
\matrix (Y3) [matrix of math nodes,anchor = north west, below = of Y2-2-5, xshift=4mm,%
ampersand replacement=\&,
column sep=10pt,
row sep = 3pt,
column 1/.style={nulIndColor,column sep=-26pt},
nodes={anchor= center,minimum height=15pt}
]
{
    \phantom{{ }}\&\phantom{{ }}\&|[color=colIndColor]|\seq{1,2} \&|[color=colIndColor]|\seq{1,3} \&|[color=colIndColor]|\seq{1,4}\\
    \seq{3,4}\& \phantom{{ }} \& \psi _{f_1,4} \psi _{f_2,3}-\psi _{f_1,3} \psi _{f_2,4} \& \psi _{f_1,2} \psi _{f_2,4}-\psi _{f_1,4} \psi _{f_2,2} \&\psi _{f_1,3} \psi _{f_2,2}-\psi _{f_1,2} \psi _{f_2,3} \& \cdots \\[10mm]
    \phantom{{ }}\&\phantom{{ }}\&|[color=colIndColor]|\seq{2,3}\&|[color=colIndColor]|\seq{2,4}\&|[color=colIndColor]|\seq{3,4}\\
    \phantom{{ }} \& \cdots \&  \psi _{f_1,4} \psi _{f_2,1}-\psi _{f_1,1} \psi _{f_2,4} \& \psi _{f_1,1} \psi _{f_2,3}-\psi _{f_1,3} \psi _{f_2,1} \& \psi _{f_1,2} \psi _{f_2,1}-\psi _{f_1,1} \psi _{f_2,2} \\
};
\node[left delimiter=[, xshift=5pt, yshift=1pt, fit={(Y3-2-3.north west)(Y3-2-3.south west)}] (brleft){};
\node[right delimiter={]},fit={($ (Y3-4-5.north east) + (0pt,0pt)$)($(Y3-4-5.south east) + (0pt,0pt)$)}] (brright){}; 
\node[left = -5pt of Y3-2-1,yshift=0mm]{$=$};
\end{tikzpicture}
\end{align*}
and show that this vector is in the left null-space of $\repMat{\{f_1,f_2\}}{2}$. The general construction of the null-space is presented in the proof of Proposition~\ref{prop:multi-beta} in Section~\ref{sec:nddproofs}. 

First note that $\NSpc{\cE}{2}$ is not an all-zero vector, otherwise we have 
\[
\frac{\psi_{f_1,1}}{\psi_{f_2,1}} = \frac{\psi_{f_1,2}}{\psi_{f_2,2}} = \frac{\psi_{f_1,3}}{\psi_{f_2,3}} = \frac{\psi_{f_1,4}}{\psi_{f_2,4}},
\]
which implies rows $\enc_{f_1}$ and $\enc_{f_2}$ of the encoder matrix are linearly dependent. This is in contradiction with the fact that every $d=4$ rows of $\enc$ are linearly independent. Hence, without loss of generality, we may assume 
\begin{align*}
\begin{vmatrix} \psi_{f_1,1} & \psi_{f_1,4} \\ \psi_{f_2,1} & \psi_{f_2,4} \end{vmatrix}\neq 0.
\end{align*} 

In order to prove $\NSpc{\cE}{2} \cdot \repMat{\cE}{2}=0$, we show that vector $\NSpc{\cE}{2}$ is orthogonal to each column of $\repMat{\cE}{2}$. For instance,  consider the seventh column of $\repMat{\cE}{2}$, labeled by $\set{3}$ in segment $\repMat{f_2}{2}$. The inner product of $\NSpc{\cE}{2}$ and this column is given by
\begin{align*}
-\psi_{f_2,1} \begin{vmatrix}\psi_{f_1,2} & \psi_{f_1,4} \\\psi_{f_2,2} & \psi_{f_2,4}\end{vmatrix} +\psi_{f_2,2} \begin{vmatrix}\psi_{f_1,1} & \psi_{f_1,4} \\\psi_{f_2,1} & \psi_{f_2,4}\end{vmatrix}-\psi_{f_2,4} \begin{vmatrix}\psi_{f_1,1} & \psi_{f_1,2} \\\psi_{f_2,1} & \psi_{f_2,2}\end{vmatrix}= - \begin{vmatrix}\psi_{f_2,1} & \psi_{f_2,2} & \psi_{f_2,4}\\ \psi_{f_1,1} & \psi_{f_1,2} & \psi_{f_1,4} \\\psi_{f_2,1} & \psi_{f_2,2} & \psi_{f_2,4}\end{vmatrix}=0,
\end{align*}
where the first equality follows from the Laplace expansion of the determinant with respect to the first row, and the second equality is due to the fact that the first and third rows of the matrix are identical, and hence it is rank-deficient. The existence of a non-zero vector in the left null-space of $\repMat{\cE}{2}$ implies its rank is at most $\beta_2^{(2)}=5$. 

Now, assume $h=1$ is one of the helper nodes. Without loss of generality, we may assume $\psi_{f_1,1}\neq 0$ and $\psi_{f_2,1} \neq 0$. The repair data sent by node $1$, i.e., $\enc_{1} \cdot \SM \cdot \repMat{\cE}{2}$, has the following $8$ symbols:
\begin{align}
\left\{
\begin{array}{l}
\mathsf{Symbol\;}1:\psi _{f_1,2} v_{1,\{1,2\}} + \psi _{f_1,3} v_{1,\{1,3\}} + \psi _{f_1,4} v_{1,\{1,4\}},\\
\mathsf{Symbol\;}2:-\psi _{f_1,1} v_{1,\{1,2\}}+ \psi _{f_1,3} w_{1,\{1,2,3\}} + \psi _{f_1,4} w_{1,\{1,2,4\}}, \\
\mathsf{Symbol\;}3:-\psi _{f_1,1} v_{1,\{1,3\}} - \psi _{f_1,2} w_{1,\{1,2,3\}}+ \psi _{f_1,4} w_{1,\{1,3,4\}}, \\
\mathsf{Symbol\;}4:-\psi _{f_1,1} v_{1,\{1,4\}} - \psi _{f_1,2} w_{1,\{1,2,4\}} - \psi _{f_1,3} w_{1,\{1,3,4\}}\\
\mathsf{Symbol\;}5:\psi _{f_2,2} v_{1,\{1,2\}} + \psi _{f_2,3} v_{1,\{1,3\}} + \psi _{f_2,4} v_{1,\{1,4\}},\\
\mathsf{Symbol\;}6:-\psi _{f_2,1} v_{1,\{1,2\}}+ \psi _{f_2,3} w_{1,\{1,2,3\}} + \psi _{f_2,4} w_{1,\{1,2,4\}}, \\
\mathsf{Symbol\;}7:-\psi _{f_2,1} v_{1,\{1,3\}} - \psi _{f_2,2} w_{1,\{1,2,3\}}+ \psi _{f_2,4} w_{1,\{1,3,4\}}, \\
\mathsf{Symbol\;}8:-\psi _{f_2,1} v_{1,\{1,4\}} - \psi _{f_2,2} w_{1,\{1,2,4\}} - \psi _{f_2,3} w_{1,\{1,3,4\}},\\
\end{array}
\right.
\label{eq:ex:rep-symbs}
\end{align}
However, we claim that $\mathsf{Symbol\;}1$, $\mathsf{Symbol\;}5$, and $\mathsf{Symbol\;}8$ are redundant, and can be reconstructed as linear combinations of the remaining five symbols. This can be verified by 
\begin{align*}
\mathsf{Symbol\;} 1 &=  -\frac{\psi _{f_1,2}}{\psi _{f_1,1}} \times \mathsf{Symbol\;}2 
 					   -\frac{\psi _{f_1,3}}{\psi _{f_1,1}} \times \mathsf{Symbol\;}3 
					   -\frac{\psi _{f_1,4}}{\psi _{f_1,1}} \times \mathsf{Symbol\;}4 
					   \\
\mathsf{Symbol\;} 5 &=  -\frac{\psi _{f_2,2}}{\psi _{f_2,1}} \times \mathsf{Symbol\;}6 
 					   -\frac{\psi _{f_2,3}}{\psi _{f_2,1}} \times \mathsf{Symbol\;}7 
					   -\frac{\psi _{f_2,4}}{\psi _{f_2,1}} \times \mathsf{Symbol\;}8. 	\\		
\mathsf{Symbol\;}8 &= \frac{\psi _{f_2,1} \left(\psi _{f_1,1} \psi _{f_2,2}-\psi _{f_1,2} \psi _{f_2,1}\right)}{\psi _{f_1,1} \left(\psi _{f_1,1} \psi _{f_2,4}-\psi _{f_1,4} \psi _{f_2,1}\right)}  \times \mathsf{Symbol\;}2
+\frac{\psi _{f_2,1} \left(\psi _{f_1,1} \psi _{f_2,3}-\psi _{f_1,3} \psi _{f_2,1}\right)}{\psi _{f_1,1} \left(\psi _{f_1,1} \psi _{f_2,4}-\psi _{f_1,4} \psi _{f_2,1}\right)} \mathsf{Symbol\;}3\\
&\phantom{=} +
\frac{\psi _{f_2,1}}{\psi _{f_1,1}} \mathsf{Symbol\;}4 +\frac{\psi _{f_1,1} \psi _{f_2,2}-\psi _{f_1,2} \psi _{f_2,1}}{\psi _{f_1,4} \psi _{f_2,1}-\psi _{f_1,1} \psi _{f_2,4}} \mathsf{Symbol\;}6+
\frac{\psi _{f_1,1} \psi _{f_2,3}-\psi _{f_1,3} \psi _{f_2,1}}{\psi _{f_1,4} \psi _{f_2,1}-\psi _{f_1,1} \psi _{f_2,4}} \mathsf{Symbol\;}7					  
\end{align*}
It is worth noting that the first and second equations above indicate the  dependencies between symbols that are sent for the repair of $f_1$ and $f_2$,  respectively (similar to the \eqref{eq:zero-lin-comb}). The third equation, however, shows an additional dependency across the repair symbols $f_1$ and and those of $f_2$. This implies that it suffices for the helper node $1$ to send symbols $2$, $3$, $4$, $6$, and $7$ in \eqref{eq:ex:rep-symbs} to repair two nodes $f_1$ and $f_2$, simultaneously.

\section{Proofs of the  Code Properties}
\label{sec:nddproofs}

In this section, we provide formal proofs for the exact-repair property of the code stated in Propositions~\ref{prop:node:rep}. We also prove the bounds on the repair bandwidth for single and multiple node failures, presented in Propositions~\ref{lm:beta} and Propositions~\ref{prop:multi-beta}, respectively. 
\begin{proof}[Proof of proposition~\ref{prop:node:rep}]
The proof technique here is similar to that used in  \eqref{eq:R:Example}. We start with the RHS of \eqref{eq:repair}, and plugin the entries of matrices  $\repSpc{f}{m}$, $\repMat{f}{m}$ and $\SM$, to expand it. Next, we split the terms in the summation into $v$-symbols ($\SM_{x,\sen I}$ with $x\in \sen I$) and $w$-symbols ($\SM_{y,\sen I}$ with $y\notin \sen I$), and then we prove the identity for $v$ and $w$ symbols separately. The details of proof are as follows.
	\begin{align}
	\sum_{x \in \sen I} (-1)^{\ind{\sen I}{x}}& \left[\repSpc{f}{m}\right]_{x,\sen I \setminus \seq{x}} = \sum_{x \in \sen I} (-1)^{\ind{\sen I}{x}} \left[\SM \cdot \repMat{f}{m} \right]_{x,\sen I \setminus \seq{x}} \nonumber\\ 
	&= \sum_{x \in \sen I} (-1)^{\ind{\sen I}{x}} \sum_{\stackrel{\sen L\subseteq \intv{d}}  {\size{\sen L}=m}}\SM_{x,\sen L}  \cdot\repMat{f}{m}_{\sen L ,\sen I \setminus \seq{x}} \nonumber\\ 
	&= \sum_{x \in \sen I} (-1)^{\ind{\sen I}{x}} \sum_{ y\in \intv{d}\setminus{(\sen I \setminus{\set{x}})} } \SM_{x,\sen (\sen I \setminus{\set{x}}) \cup \set{y}}  \cdot \repMat{f}{m}_{\sen (\sen I \setminus{\set{x}}) \cup \set{y},\sen I \setminus \seq{x}} \label{eq:rep:iden:3}\\ 
	&= \sum_{x \in \sen I} (-1)^{\ind{\sen I}{x}} \left[ \SM_{x,\sen I}   \repMat{f}{m}_{\sen I ,\sen I \setminus \seq{x}} + \sum_{{y\in \intv{d}\setminus{\sen I}} }\SM_{x,\sen (\sen I \setminus{\set{x}}) \cup \set{y}}  \cdot\repMat{f}{m}_{\sen (\sen I \setminus{\set{x}}) \cup \set{y} ,\sen I \setminus \seq{x}}\right] \label{eq:rep:iden:4}\\ 
	&= \sum_{x \in \sen I} (-1)^{\ind{\sen I}{x}} \left[ (-1)^{\ind{\sen I}{x}} \psi_{f,x}\SM_{x,\sen I} + \sum_{y\in \intv{d}\setminus{\sen I} }(-1)^{\ind{(\sen I \setminus \seq{x})\cup \set{y}}{y} }\psi_{f,y} \SM_{x,\sen (\sen I \setminus{\set{x}}) \cup \set{y}}\right] \label{eq:rep:iden:6}  \\ 
	&= \sum_{x \in \sen I} \psi_{f,x}\SM_{x,\sen I} + \sum_{y\in \intv{d}\setminus{\sen I} } \sum_{x \in \sen I}  (-1)^{\ind{\sen I}{x}} (-1)^{\ind{(\sen I \setminus \seq{x})\cup \set{y}}{y} }  \psi_{f,y} \SM_{x,\sen (\sen I \setminus{\set{x}}) \cup \set{y}}   \label{eq:rep:iden:7}\\ 
	&= \sum_{x \in \sen I} \psi_{f,x}\SM_{x,\sen I} + \sum_{y\in \intv{d}\setminus{\sen I}} \sum_{x \in \sen I}  \psi_{f,y} (-1)^{\ind{\sen I \cup \set{y}}{y}+1}  (-1)^{\ind{\sen I \cup \set{y}}{x}}  \SM_{x,\sen (\sen I \setminus{\set{x}}) \cup \set{y}} \label{eq:rep:iden:8}  \\ 
    &= \sum_{x \in \sen I} \psi_{f,x}\SM_{x,\sen I} +\sum_{y\in \intv{d}\setminus{\sen I} }(-1)^{\ind{\sen I \cup \set{y}}{y}+1} \psi_{f,y}\sum_{x \in \sen I} (-1)^{\ind{\sen I \cup \set{y}}{x}}  w_{x,\sen I \cup \set{y}}  \label{eq:rep:iden:10} \\
    &= \sum_{x \in \sen I} \psi_{f,x}\SM_{x,\sen I} +\sum_{y\in \intv{d}\setminus{\sen I} }(-1)^{\ind{\sen I \cup \set{y}}{y}+1} \psi_{f,y} \left[ (-1)^{\ind{\sen I \cup \set{y}}{y}+1}  w_{y,\sen I \cup \set{y}} \right] \label{eq:rep:iden:11} \\
	&= \sum_{x \in \sen I} \psi_{f,x}\SM_{x,\sen I} +\sum_{y\in \intv{d}\setminus{\sen I} }(-1)^{\ind{\sen I \cup \set{y}}{y}+1} \psi_{f,y} \left[(-1)^{\ind{\sen I \cup \set{y}}{y}+1}  \SM_{y,\sen I } \right]   \label{eq:rep:iden:12} \\
	&= \sum_{x \in \sen I} \psi_{f,x}\SM_{x,\sen I} +\sum_{y\in \intv{d}\setminus{\sen I} }\psi_{f,y} \SM_{y,\sen I} \label{eq:rep:iden:13} \\
	&= \sum_{x \in \intv{d}} \psi_{f,x}\SM_{x,\sen I} =
	\left[\psi_f \SM \right]_{\sen I}
	\end{align}
where in above equations we used the following facts.
\begin{itemize}
	\item In~\eqref{eq:rep:iden:3}, we used the definition of $\repMat{f}{m}$ in \eqref{eq:def:xi} which implies  $\repMat{f}{m}_{\sen L, \sen I \setminus\set{x}}$ is non-zero only if $\sen L$ includes $\sen I \setminus\set{x}$;
    \item In~\eqref{eq:rep:iden:4}, we split the summation into two cases:  $y=x$ and  $ y \neq x$;
	\item In~\eqref{eq:rep:iden:6}, we replaced  $\repMat{f}{m}_{\sen I ,\sen I \setminus \seq{x}}$ by $(-1)^{\ind{\sen I}{x}} \psi_{f,x}$ from its definition in \eqref{eq:def:xi}; 
 \item In~\eqref{eq:rep:iden:7}, the two summations over $x$ and $y$ are swapped;  
   \item In~\eqref{eq:rep:iden:8}, we used the identity $(-1)^{\ind{(\sen I \setminus \seq{x})\cup \set{y}}{y} }  (-1)^{\ind{\sen I}{x}}= (-1)^{\ind{\sen I \cup \set{y}}{y}+1}  (-1)^{\ind{\sen I \cup \set{y}}{x}}$. In order to prove the identity, we may consider two cases: 
   \begin{itemize}
       \item If $x<y$, then
        \begin{align*}
		\left\{        
		\begin{array}{l}
        \ind{(\sen I \setminus \seq{x})\cup \set{y}}{y} = \ind{\sen I \cup \set{y}}{y}-1, \\
        \ind{\sen I}{x} =  \ind{\sen I \cup \set{y}}{x}. 
		\end{array}
		\right.
        \end{align*}
       \item If $x > y$, then 
        \begin{align*}
		\left\{        
		\begin{array}{l}
        \ind{(\sen I \setminus \seq{x})\cup \set{y}}{y} = \ind{\sen I \cup \set{y}}{y}, \\
        \ind{\sen I}{x} =  \ind{\sen I \cup \set{y}}{x}-1. 
		\end{array}
		\right.
        \end{align*}       
   \end{itemize}
   \item In~\eqref{eq:rep:iden:10}, since $x\notin (\sen I\setminus \set{x}) \cup \set{y}$ then $\SM_{x,(\sen I\setminus \set{x}) \cup \set{y}}=w_{x,\sen I \cup \set{y}}$;
   \item In~\eqref{eq:rep:iden:11}, we used the parity equation~\eqref{eq:parityeq}. In particular, we have $\sum_{x \in \sen I \cup \set{y}}(-1)^{\ind{\sen I \cup \set{y}}{x}} w_{x,I \cup \set {y}} = 0$,  which implies  $\sum_{x \in \sen I } (-1)^{\ind{\sen I \cup \set{y}}{x}}w_{x,I \cup \set {y}}=- (-1)^{\ind{\sen I \cup \set{y}}{y}}w_{y,I \cup \set {y}}$. 
\end{itemize}
This completes the proof. 
 \end{proof}

\begin{proof}[Proof of Proposition~\ref{lm:beta}]
      In order to show that the repair bandwidth constraint is fulfilled, we need to show  that the rank of matrix $\repMat{f}{m}$ is at most $\beta^{(m)}=\binom{d-1}{m-1}$. First, note that it is easy to verify the claim for $m=1$, since the matrix $\repMat{f}{1}$   has only one column labeled by $\varnothing$ and hence its rank  is at most $1=\binom{d-1}{1-1}$. For $m>1$, we  partition the columns of the matrix into $2$ disjoint groups of size $\beta^{(m)}=\binom{d-1}{m-1}$ and $\binom{d}{m-1}-\beta^{(m)} = \binom{d-1}{m-2}$, and show that each column in the second group can be written as a linear combination of the columns in the first group. This implies that the rank of the matrix does not exceed the number of columns in the first group, which is exactly $\beta^{(m)}$. 
   
To form the groups, we pick some $x\in \intv{d}$ such that\footnote{Note that such an $x$ exists, otherwise the $f$-th row $\enc$ will be zero, which is in contradiction with the fact that every $d$ rows of $\enc$ are linearly independent.} $\psi_{f,x}\neq 0$. Then the first group is the set of all columns whose label is a subset of $\intv{d}\setminus\{x\}$. Recall that columns of $\repMat{f}{m}$ are labeled with $(m-1)$-element subsets of $\intv{d}$. Hence, the number of columns in the first group is $\binom{d-1}{m-1}$. Then, the second group is formed by those columns for which $x$ appears in their label. 

Without loss of generality, we may assume $x=d$, i.e $\psi_{f,d}\neq 0$, and hence the first group consists of columns $\sen I$ such that $\sen I \subset \intv{d-1}$, and the second group includes those $\sen I$'s such that $d\in \sen I$.  For every  $\sen I$ with $d\in \sen I$, we claim that
    \begin{align}
\repMat{f}{m}_{\bLozenge,\sen I } = (-1)^m \psi_{f,d}^{-1} \sum_{y\in \intv{d-1} \setminus \sen J} (-1)^{\ind{\sen J \cup \set{y}}{y}} \psi_{f,y} \repMat{f}{m}_{\bLozenge,\sen J \cup \set{y}},
\label{eq:claim:lm1}
    \end{align} 
    where $\sen J = \sen I \setminus \{d\}$. Note that all the columns appear in the RHS of \eqref{eq:claim:lm1} belong to the first group. Given the facts that $|\sen I|=m-1$ and $\ind{\sen I}{d} = m-1$, the equation in \eqref{eq:claim:lm1} is equivalent to
    \begin{align}
    \sum_{y\in \intv{d} \setminus \sen J} (-1)^{\ind{\sen J \cup \set{y}}{y}} \psi_{f,y} \repMat{f}{m}_{\bLozenge,\sen J \cup \set{y}} = 0.
    \label{eq:claim:lm2}
    \end{align}
Let us focus on an arbitrarily chosen row of the matrix, labeled by $\sen L$, where $\sen L \subseteq \intv{d}$ with $|\sen L|=m$. The $\sen L$-th entry of the column in the LHS of \eqref{eq:claim:lm2} is given by 
    \begin{align*}
    \left[\sum_{y\in \intv{d} \setminus \sen J} (-1)^{\ind{\sen J \cup \set{y}}{y}} \psi_{f,y} \repMat{f}{m}_{\bLozenge,\sen J \cup \set{y}}\right]_{\sen L} &=\sum_{y\in \intv{d} \setminus \sen J} (-1)^{\ind{\sen J \cup \set{y}}{y}} \psi_{f,y} \repMat{f}{m}_{\sen L,\sen J \cup \set{y}}.
    \end{align*}
First assume $\sen J \not\subset \sen L$. This together with the definition of $\repMat{f}{m}$ in \eqref{eq:def:xi} imply that  $\repMat{f}{m}_{\sen L,\sen J \cup \set{y}}=0$ for any $y$, and hence all the terms in the LHS of \eqref{eq:claim:lm2} are zero. 

Next, consider an $\sen L$ such that $\sen J \subseteq \sen L$. Since $|\sen J|=|\sen I \setminus \set{d}|=m-2$ and $|\sen L|=m$, we have $\sen L = \sen J \cup \{y_1, y_2\}$, where $y_1<y_2$ are  elements of $\intv{d}$. Note that for $y\notin \{y_1,y_2\}$ we have $\repMat{f}{m}_{\sen L,\sen J \cup \set{y}}=0$, since $\sen J \cup \set{y} \not\subset \sen L$. Therefore,  \eqref{eq:claim:lm2} can be simplified as
    \begin{align}
    \Bigg[\sum_{y\in \intv{d} \setminus \sen J} & (-1)^{\ind{\sen J \cup \set{y}}{y}} \psi_{f,y} \repMat{f}{m}_{\bLozenge,\sen J \cup \set{y}}\Bigg]_{\sen L} =\sum_{y\in \{y_1, y_2\}} (-1)^{\ind{\sen J \cup \set{y}}{y}} \psi_{f,y} \repMat{f}{m}_{\sen L,\sen J \cup \set{y}}\label{eq:claim:lm3-1} \\
    &= (-1)^{\ind{\sen J\cup \set{y_1}}{y_1}}\psi_{f,y_1} \repMat{f}{m}_{\sen L,\sen J \cup \set{y_1}}+(-1)^{\ind{\sen J\cup \set{y_2}}{y_2}}\psi_{f,y_2} \repMat{f}{m}_{\sen L,\sen J \cup \set{y_2}}\nonumber\\
    &=(-1)^{\ind{\sen J\cup \set{y_1}}{y_1}}\psi_{f,y_1}\cdot (-1)^{\ind{\sen L}{y_2}} \psi_{f,y_2}+(-1)^{\ind{\sen J\cup \set{y_2}}{y_2}}\psi_{f,y_2}\cdot (-1)^{\ind{\sen L}{y_1}} \psi_{f,y_1}
    \label{eq:claim:lm3-2}\\
    &=\psi_{f,y_1} \psi_{f,y_2} \left[(-1)^{\ind{\sen L\setminus \set{y_2}}{y_1}+\ind{\sen L}{y_2}}+(-1)^{\ind{\sen L\setminus \set{y_1}}{y_2}+\ind{\sen L}{y_1}}\right]\nonumber\\
    &=0, \label{eq:claim:lm3-3}  
    \end{align}
where \eqref{eq:claim:lm3-1} and \eqref{eq:claim:lm3-2} is due to the definition of $\repMat{f}{m}$ in \eqref{eq:def:xi}, and in \eqref{eq:claim:lm3-3}  we used the facts that $\ind{\sen L\setminus \set{y_2}}{y_1}= \ind{\sen L}{y_1}$ and $\ind{\sen L\setminus \set{y_1}}{y_2}= \ind{\sen L}{y_2}-1$ which holds for $y_1,y_2\in \sen L$ with $y_1<y_2$. This completes the proof. 
\end{proof}

\begin{proof}[Proof of Proposition~\ref{prop:multi-beta}]
The claim of this proposition for $e>d$ is equivalent to bounding the number of repair symbols from the helper node by $\alpha=\binom{d}{m}$ (because $\binom{d-e}{m}=0$), which is clearly true since each storage node does not store more than $\alpha$ symbols. Thus, we can limit our attention to  $e\leq d$. In order to prove the claim, we show that the row rank of $\repFam{\sen E}{m}$ does not exceed $\binom{d}{m} - \binom{d-e}{m}$. Recall that $\repFam{\sen E}{m}$ is a $\binom{d}{m} \times e\binom{d}{m-1}$ matrix, and it suffices to identify $\binom{d-e}{m}$ linearly independent vectors in the \emph{left null-space} of $\repFam{\sen E}{m}$. To this end, we introduce a full-rank matrix $\NSpc{\sen E}{m}$ of size $\binom{d-e}{m} \times \binom{d}{m}$ and show that $\NSpc{\sen E}{m} \cdot  \repFam{\sen E}{m} = \mathbf{0}$.

\textbf{Step 1 (Construction of $\NSpc{\sen E}{m}$):} 
Let $\enc[\sen E, :]$ be the $e \times d$ matrix obtained from the rows $f_1,f_2,\cdots,f_e$ of the encoder matrix $\enc$. Recall that $\enc[\sen E, :]$ is full-rank (since any $d$ rows of $\enc$ are linearly independent, and $e\leq d$). Hence, there exists a subset $Q$ with $|\sen Q| = e$ of the columns of $\enc[\sen E, :]$, denoted by $\enc[\sen E,\sen Q]$ such that $\det{\enc[\sen E, \sen Q]} \neq 0$. 

The desired matrix $\NSpc{\sen E}{m}$ is of size $\binom{d-e}{m} \times \binom{d}{m}$. We label its rows by $m$-element subsets of $\intv{d} \setminus \sen Q$, and its columns by $m$-element subsets of $\intv{d}$. Then the entry at row $\sen I$ and column $\sen L$ is defined as
		\begin{align}
		\NSpc{\sen E}{m}_{\sen I,\sen L} = \begin{cases}
		(-1)^{\sum_{j\in \sen L} \ind{\sen I \cup \sen Q}{j}}\det{\enc[\sen E, \left(\sen I \cup \sen Q \right) \setminus \sen L]} & \textrm{if $\sen L \subseteq \sen I \cup \sen Q$}, \\
		0 & \textrm{if $\sen L\nsubseteq \sen I \cup \sen Q$}.
		\end{cases} \label{eq:NSpc:def}
		\end{align}
		Note that $\sen I \subseteq \intv{d} \setminus \sen Q$, and hence $|\sen I \cup \sen Q|=m+e$. 

 \textbf{Step 2 (Orthogonality of $\NSpc{\sen E}{m}$ to $\repFam{\sen E}{m}$):} We prove this claim for each segment of $\repFam{\sen E}{m}$. More precisely, for each $f\in \sen E$, we prove $\NSpc{\sen E}{m} \cdot \repMat{f}{m} = \mathbf{0}$. Consider some $f\in \sen E$, and arbitrary indices $\sen I$ and $\sen J$ for  rows and columns, respectively. We have 
		\begin{align}
			\left[\NSpc{\sen E}{m} \cdot \repMat{f}{m}\right]_{\sen I,\sen J} &= \sum_{\substack{ \sen L \subseteq \intv d \\ \size{\sen L}=m}}\NSpc{\sen E}{m}_{\sen I ,\sen L} \cdot \repMat{f}{m}_{\sen L,\sen J}\\
			&= \sum_{x  \in \intv{d} \setminus \sen J }\NSpc{\sen E}{m}_{\sen I,\sen J \cup \set{x} } \cdot\repMat{f}{m}_{\sen J \cup \set{x},\sen J}\label{eq:xirank:1}\\
			&= \sum_{x  \in \left( \sen I \cup \sen Q \right) \setminus \sen J }\NSpc{\sen E}{m}_{\sen I,\sen J \cup \set{x} } \cdot\repMat{f}{m}_{\sen J \cup \set{x},\sen J}\label{eq:xirank:1-2}\\
			&= \sum_{x  \in \left( \sen I \cup \sen Q \right) \setminus \sen J  } \left[(-1)^{\sum_{j\in \sen J \cup \set{x}} \ind{\sen I \cup \sen Q}{j} } \det{\enc[\sen E, \left(\sen I \cup \sen Q \right) \setminus \left(\sen J \cup \set{x}\right)]} \right] \left[(-1)^{\ind{\sen J \cup \set{x}}{x}} \psi_{f,x}\right]\label{eq:xirank:3}\\
			&= (-1)^{\sum_{j\in \sen J} \ind{\sen I \cup \sen Q}{j}-1} \sum_{x  \in \left( \sen I \cup \sen Q \right) \setminus \sen J  } (-1)^{\ind{\left(\sen I \cup \sen Q \right) \setminus \sen J}{x}} \det{\enc[\sen E, \left(\sen I \cup \sen Q \right) \setminus \left(\sen J \cup \set{x}\right)]} \psi_{f,x}\label{eq:xirank:4}\\			
					&= (-1)^{\sum_{j\in \sen J} \ind{\sen I \cup \sen Q}{j} -1} \quad \det{ \left[ \begin{array}{c}\enc[f, \left(\sen I \cup \sen Q\right) \setminus\sen J] \\ \hline \enc[\sen E,\left(\sen I \cup \sen Q \right)\setminus \sen J] \end{array} \right]}\label{eq:xirank:5}\\
			&=0. \label{eq:xirank:6}
		\end{align}
Note that
	\begin{itemize}
 \item In \eqref{eq:xirank:1} we have used the fact that  $\repMat{f}{m}$ is non-zero only if  $ \sen L= \sen J  \cup  \set{x}$ for some $x\in \intv{d}\setminus \sen J$ (see  \eqref{eq:def:xi}). 

\item  The range of $x$ is further limited in \eqref{eq:xirank:1-2} due to the fact that $\NSpc{\sen E}{m}_{\sen I,\sen J \cup \set{x} }$ is non-zero only if $\sen J \cup \set{x} \subseteq \sen I \cup \sen Q$  (see \eqref{eq:NSpc:def}), which implies $x \in \left(\sen I \cup \sen Q \right) \setminus \sen J$.
\item In \eqref{eq:xirank:3}, the entries of the matrix product are replaced by their definitions from \eqref{eq:def:xi} and \eqref{eq:NSpc:def}.
		
\item The equality in~\eqref{eq:xirank:4} is obtained by factoring $(-1)^{\sum_{j\in \sen J} \ind{\sen I \cup \sen Q}{j}}$, and using the facts that $\sen J \subset \sen I \cup \sen Q$ and 
\begin{align*}
\ind{\sen I \cup \sen Q}{x} + \ind{\sen J \cup \{x\}}{x} &\equiv 
\ind{\sen I \cup \sen Q}{x} - \ind{\sen J \cup \{x\}}{x} \qquad \qquad (\mathsf{mod\;\;} 2) \nonumber\\
&= 
\left|\left\{y\in \sen I \cup \sen Q: y\leq x\right\}\right| - \left|\left\{y\in \sen J \cup \{x\}: y\leq x\right\}\right|\nonumber\\
&=
\left|\left\{y\in \sen I \cup \sen Q: y\leq x\right\}\right| - \left|\left\{y\in \sen J : y\leq x\right\}\right| -1\nonumber\\
&= \left|\left\{y\in (\sen I \cup \sen Q)\setminus \sen J: y\leq x\right\}\right|  -1
= \ind{(\sen I \cup \sen Q)\setminus \sen J}{x}  -1.
		\end{align*}

		\item The equality in~\eqref{eq:xirank:5} follows the determinant expansion of the matrix with respect to its first row. Note that $|(\sen I \cup \sen Q)\setminus \sen J|=e+1$, and hence it is a square matrix. 
		
		\item Finally, the determinant in ~\eqref{eq:xirank:5} is zero, because $f\in \sen E$, and hence the matrix has two identical rows. 
	\end{itemize}

 \textbf{Step 3 (Full-rankness of $\NSpc{\sen E}{m}$):} Recall that rows and columns of $\NSpc{\sen E}{m}$ are labeled by $m$-element subsets of  $\intv{d}\setminus \sen Q$ and $m$-element subsets of $\intv{d}$, respectively. Consider the sub-matrix of $\NSpc{\sen E}{m}$, whose column labels are subsets of $\intv{d}\setminus \sen Q$. This is $\binom{d-e}{m}\times \binom{d-e}{m}$ square sub-matrix. Note that for entry at position $(\sen I, \sen J)$ with $\sen I \neq \sen J$, since $\sen J \cap \sen Q = \emptyset$, we have $\sen J \nsubseteq \sen I \cup \sen Q$, and hence $\left[\NSpc{\sen E}{m}\right]_{\sen I, \sen J} = 0$  (see \eqref{eq:NSpc:def}). Otherwise, if $\sen I = \sen J$ we have $\left[\NSpc{\sen E}{m}\right]_{\sen I, \sen J} = (-1)^{\sum_{i\in \sen I} \ind{\sen I \cup \sen Q}{i}} \det{\enc[\sen E, \sen Q]}$. That is
	\begin{align*}
	 	\left[\NSpc{\sen E}{m}\right]_{\sen I, \sen J} = \begin{cases} 
	 	(-1)^{\sum_{i\in \sen I} \ind{\sen I \cup \sen Q}{i}} \det{\enc[\sen E, \sen Q]} & \textrm{if $\sen I = \sen J$},\\
	 	0 & \textrm{if  $\sen I \neq \sen J$}.
	 	\end{cases} 
	\end{align*}
This implies that $\NSpc{\sen E}{m}$ has a diagonal sub-matrix, with  diagonal entries  $\pm \det{\enc[\sen E, \sen Q]}$ which are non-zero (see \textbf{Step 1}), and hence  $\NSpc{\sen E}{m}$ is a full-rank matrix. Therefore, the rows of $\NSpc{\sen E}{m}$ provide  $\binom{d-e}{m}$ linearly independent vectors in the left null-space of $\repMat{\sen E}{m}$, and thus the rank of $\repMat{\sen E}{m}$ does not exceed $\beta^{(m)} = \binom{d}{m}-\binom{d-e}{m}$. This completes the proof. 
\end{proof}

\section{Improved Centralized Repair for Multiple Failures}
\label{sec:MulRep-improved}
In this section, we prove Theorem~\ref{thm:MulRep-improved}. As mentioned before, this result is essentially obtained from Theorem~\ref{thm:MulRep}, by exploiting the fact that in the centralized repair setting once one failed node is repaired, it can also participate in the repair process of the remaining failed nodes. 

\begin{proof}[Proof of  Theorem~\ref{thm:MulRep-improved}]
Consider a set of $e$ failed nodes  $\mathcal{E} = \set{f_1,f_2,\cdots, f_e}$, which are going to be repaired by a set of helper nodes $\cH$ with $|\cH|=d$. Recall from Theorem~\ref{thm:MulRep} that the repair data of node $h$ intended for a failed node $f$ (i.e., $\enc_h \cdot \SM \cdot \repMat{f}{m}$) can be retrieved from the repair data that $h$ sends for the repair of a set of failed nodes $\mathcal{E}$ (i.e., $\enc_h \cdot \SM \cdot \repMat{\sen E}{m}$) where $f\in \mathcal{E}$. We use the following procedure in order to repair the failed nodes in $\sen E$. 
\begin{enumerate}
	\item First node $f_1$ is repaired using helper nodes $\cH=\{h_1, h_2, \cdots, h_d\}$. 
	\item Having failed $\{f_1,\dots, f_{i}\}$ repaired, the repair process of failed node $f_{i+1}$ is performed using helper nodes $\cH_i = \{f_1,\dots, f_{i}\} \cup \{h_{i+1},h_{i+2},\dots, h_{d}\}$. This step will be repeated for $i=2,3,\dots, e$. 
\end{enumerate}

Note that the proposed repair process introduced in this paper is helper-independent, and hence the repair data sent to a failed node $f_i$ by a helper node $h$ does not depend on the identity of the other helper nodes. 

Using the procedure described above, helper node $h_i$  only participates in the repair of failed nodes $\{f_1,f_2,\dots, f_i\}$, for $i=1,2,\dots, e$, while the other helper nodes (i.e., $h_i$ for $i=e+1,\dots, d$) contribute in the repair of all the $e$ failed nodes. Hence, the total repair data downloaded from the helper nodes in $\cH$ to the central repair unit is given by 
\begin{align}
	\sum_{j=1}^{e} \beta_{j}^{(m)} +  (d-e) \beta_e^{(m)} &=   \sum_{j=1}^{e} \left[\binom{d}{m}-\binom{d-j}{m} \right] + (d-e) \left[\binom{d}{m}-\binom{d-e}{m} \right]\\
	 &= d \binom{d}{m}-\sum_{j=1}^{e} \binom{d-j}{m} -(d-e)\binom{d-e}{m}\\
	 &= d \binom{d}{m}-\left[\binom{d}{m+1}-\binom{d-e}{m+1} \right] - (d-e)\binom{d-e}{m} \label{eq:iden:x1}\\
	 &= (d+1) \binom{d}{m}-\left[\binom{d}{m}+\binom{d}{m+1}\right]-(d-e+1)\binom{d-e}{m}+\left[\binom{d-e}{m}+\binom{d-e}{m+1}\right]\\
	 &= (m+1) \binom{d+1}{m+1}-\binom{d+1}{m+1}-(m+1)\binom{d-e+1}{m+1}+\binom{d-e+1}{m+1}\label{eq:iden:x2}\\
	 &=m\left[\binom{d+1}{m+1}-\binom{d-e+1}{m+1}\right],
\end{align}
where 
\begin{itemize}
\item in \eqref{eq:iden:x1} we used the identity 
\begin{align*}
\sum_{j=a}^{b} \binom{j}{m} = \sum_{j=m}^{b} \binom{j}{m} - \sum_{j=m}^{a-1} \binom{j}{i}=\binom{b+1}{m+1}-\binom{a}{m+1},
\end{align*}
\item and the equality in  \eqref{eq:iden:x2} holds due to the Pascal identity  $\binom{a}{m} + \binom{a}{m+1} = \binom{a+1}{m+1}$, and the fact that $(a+1)\binom{a}{m} = (m+1) \binom{a+1}{m+1}$. 
\end{itemize}

Hence the average (per helper node) repair bandwidth is given by 
\begin{align}
 \bar{\beta}_e^{(m)} = \frac{1}{d} \left[m\binom{d+1}{m+1}-m\binom{d-e+1}{m+1}\right].
\end{align}

Note that the repair strategy described here is asymmetric, i.e., the repair bandwidth of helper nodes are different. However, it can be simply symmetrized by concatenating $d$ copies of the code to form a super code 
\begin{align*}
\sen C = \enc \cdot \begin{bmatrix}
\SM^{[1]}, \SM^{[2]}, \dots, \SM^{[d]}
\end{bmatrix}.
\end{align*}

\begin{figure}
\centering
\includegraphics[width=0.8\textwidth]{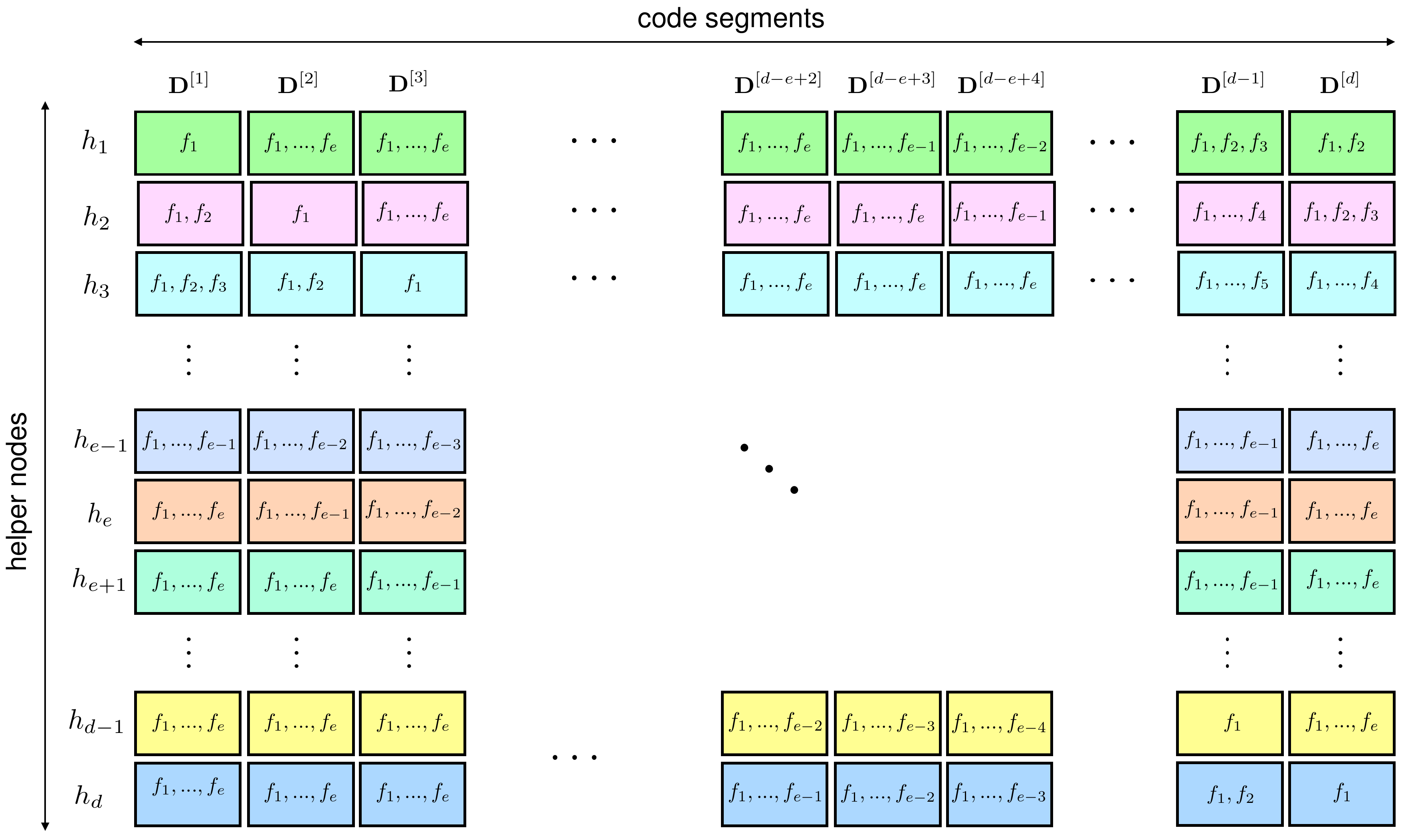}
\caption{The participation of the helper nodes in the multiple failure repair of the super-code. Each cell labeled by $h_i$ and $\SM^{[\ell]}$ shows the set of failed nodes receive repair data from helper node $h_i$ to repair their codeword segment corresponding to code segment $\SM^{[\ell]}$.} 
\label{fig:super-code}
\end{figure}
In the super code, each node stores a total of $d\cdot \alpha^{(m)}$ symbols, including $\alpha$ symbols for each code segment $\SM^{[\ell]}$, and the total storage capacity of the code is $d\cdot F^{(m)}$. In a multiple failure scenario with failed nodes $\sen E=\{f_1,\dots, f_e\}$ and helper nodes $\sen H=\{h_1,\dots, h_d\}$, codeword segments will be repaired separately. The role of helper nodes in the repair process changes in a circular manner, as shown in Fig.~\ref{fig:super-code}. Then, the per-node repair bandwidth of the super code is exactly $d\cdot \bar{\beta}_e^{(m)}$ defined above. Note that the symmetry in the super code is obtained at the expense of the sub-packetization, which is scaled by a factor of $d$. This completes the proof. 

\end{proof}


\end{document}